\RequirePackage{tikz}
\documentclass[sn-mathphys-num]{sn-jnl}% Math and Physical Sciences Numbered Reference Style
\usepackage{fixmath}
\usepackage{bm}
\usepackage{amsbsy}
\usepackage{amssymb}
\usepackage{amsthm}
\usepackage{amsmath}
\usepackage{mathtools}
\usepackage{verbatim}
\usepackage{multirow}
\usepackage{array}
\usepackage{graphicx}
\usepackage{tikz}
\usetikzlibrary{positioning}
\usepackage{color}
\usepackage{xcolor}

\usepackage{thm-restate}

\newcommand{\size}[1]{\ensuremath{|#1|}}
\newcommand{\ceil}[1]{\ensuremath{\lceil#1\rceil}}
\newcommand{\floor}[1]{\ensuremath{\lfloor#1\rfloor}}

\newcommand{\Ceil}[1]{\ensuremath{\left\lceil#1\right\rceil}}
\newcommand{\Floor}[1]{\ensuremath{\left\lfloor#1\right\rfloor}}

\newcommand{\lra}[1]{\ensuremath{(#1)}}

\newcommand{\lrA}[1]{\ensuremath{\left(#1\right)}}

\newcommand{\lrC}[1]{\ensuremath{\left\{#1\right\}}}

\let\epsilon=\varepsilon

\def\OPT{\mbox{OPT}}

\def\C{\mathcal{C}}

\def\T{\mathcal{T}}

\def\F{\mathcal{F}}

\def\I{\mathcal{I}}

\def\W{\mathcal{W}}

\def\OPT{\mbox{OPT}}

\def\C{\mathcal{C}}

\usepackage{algorithm}
\usepackage{algorithmicx}
\usepackage{algpseudocode}

\newcommand{\PP}[1]{\ensuremath{\mbox{Pr}[#1]}}

\newcommand{\EE}[1]{\ensuremath{\mathbb{E}[#1]}}

\newtheorem{theorem}{Theorem}
\newtheorem{lemma}[theorem]{Lemma}

\newtheorem{definition}[theorem]{Definition}
\newtheorem{assumption}[theorem]{Assumption}

\newtheorem{claim}[theorem]{Claim}

\begin{document}

\title[MCVRP with Improved Approximation Guarantees]{Multidepot Capacitated Vehicle Routing with Improved Approximation Guarantees}


\author{\fnm{Jingyang} \sur{Zhao} }\email{jingyangzhao1020@gmail.com}

\author{\fnm{Mingyu} \sur{Xiao}}\email{myxiao@uestc.edu.cn}
\equalcont{Corresponding author}

\affil{\orgdiv{School of Computer Science and Engineering}, \orgname{University of Electronic Science and Technology of China}, \orgaddress{\street{2006 Xiyuan Ave}, \city{Chengdu}, \postcode{610054}, \state{Sichuan}, \country{China}}}

\abstract{
The Multidepot Capacitated Vehicle Routing Problem (MCVRP) is a well-known variant of the classic Capacitated Vehicle Routing Problem (CVRP), where we need to route capacitated vehicles located in multiple depots to serve customers' demand such that each vehicle must return to the depot it starts, and the total traveling distance is minimized. There are three variants of MCVRP according to the property of the demand: unit-demand, splittable and unsplittable. We study approximation algorithms for $k$-MCVRP in metric graphs, where $k$ is the capacity of each vehicle.
The best-known approximation ratios for the three versions are $4-\Theta(1/k)$, $4-\Theta(1/k)$, and $4$, respectively.
We give a $\lra{4-1/1500}$-approximation algorithm for unit-demand and splittable $k$-MCVRP, and a $(4-1/50000)$-approximation algorithm for unsplittable $k$-MCVRP.
When $k$ is a fixed integer, we give a $(3+\ln2-\max\{\Theta(1/\sqrt{k}),1/9000\})$-approximation algorithm for the splittable and unit-demand cases, and a $(3+\ln2-\Theta(1/\sqrt{k}))$-approximation algorithm for the unsplittable case.
%Our results are based on recent progress in approximating CVRP.
}




\maketitle

\section{Introduction}
In the Multidepot Capacitated Vehicle Routing Problem (MCVRP), we are given a complete undirected graph $G=(V\cup D, E)$ with an edge weight function $w$ satisfying the symmetric and triangle inequality properties.
The $n$ nodes in $V=\{v_1,\dots,v_n\}$ represent $n$ customers and each customer $v\in V$ has a demand $d(v)\in\mathbb{Z}_{\geq 1}$.
The $m$ nodes in $D=\{u_1,\dots,u_m\}$ represent $m$ depots, with each containing an infinite number of vehicles with a capacity of $k\in\mathbb{Z}_{\geq 1}$ (we can also think that each depot contains only one vehicle, which can be used many times).
A tour is a walk that begins and ends at the same depot and the sum of deliveries to all customers in it is at most $k$.
The traveling distance of a tour is the sum of the weights of edges in the tour.
In MCVRP, we wish to find a set of tours to satisfy every customer's demand with a minimum total distance of all the tours.
%We use $k$-CVRP to denote the problem where the capacity $k$ is a constant.
In the \emph{unsplittable} version of the problem, each customer's demand can only be delivered by a single tour.
In the \emph{splittable} version, each customer's demand can be delivered by more than one tour.
Moreover, if each customer's demand is a unit, it is called the \emph{unit-demand} version.

In logistics, MCVRP is an important model that has been studied extensively in the literature (see \cite{montoya2015literature} for a survey). If there is only one depot, MCVRP is known as the famous Capacitated Vehicle Routing Problem (CVRP), and hence $k$-MCVRP is APX-hard for any fixed $k\geq3$~\cite{AsanoKTT97+}.

Next, we survey approximation algorithm for $k$-CVRP and $k$-MCVRP.

\textbf{$k$-CVRP}. Haimovich and Kan proposed~\cite{HaimovichK85} a well-known algorithm based on a given Hamiltonian cycle, called \emph{Iterated Tour Partitioning (ITP)}.
Given an $\alpha$-approximation algorithm for metric TSP, for splittable and unit-demand $k$-CVRP, ITP can achieve a ratio of $\alpha+1-\alpha/k$~\cite{HaimovichK85}.
For unsplittable $k$-CVRP, Altinkemer and Gavish~\cite{altinkemer1987heuristics} proposed a modification of ITP, called UITP, that can achieve a ratio of $\alpha+2-2\alpha/k$ for even $k$.
When $k=\infty$, $k$-CVRP becomes metric TSP.
For metric TSP, there is a well-known $3/2$-approximation algorithm~\cite{christofides1976worst,serdyukov1978some}, and Karlin \emph{et al.}~\cite{KarlinKG21,DBLP:conf/ipco/KarlinKG23} has slightly improved the ratio to $3/2-10^{-36}$. Recently, some progress has been made in approximating $k$-CVRP. Blauth \emph{et al.}~\cite{blauth2022improving} improved the ratio to $\alpha+1-\epsilon$ for splittable and unit-demand $k$-CVRP and to $\alpha+2-2\epsilon$ for unsplittable $k$-CVRP, where $\epsilon$ is a value related to $\alpha$ and satisfies $\epsilon\approx1/3000$ when $\alpha=3/2$. Then, for unsplittable $k$-CVRP, Friggstad \emph{et al.}~\cite{uncvrp} further improved the ratio to $\alpha+1+\ln2-\varepsilon$ based on an LP rounding method, where $\epsilon\approx1/3000$ is the improvement based on the method in~\cite{blauth2022improving}.
There are other improvements for the case that $k$ is a small fixed integer. Bompadre \emph{et al.}~\cite{BompadreDO06} improved the classic ratios by a term of ${\Omega}\lra{1/k^3}$ for all three versions. Zhao and Xiao~\cite{zhaocvrp} proposed a $(5/2-\Theta(1/\sqrt{k}))$-approximation algorithm for splittable and unit-demand $k$-CVRP and a $(5/2+\ln2-\Theta(1/\sqrt{k}))$-approximation algorithm for unsplittable $k$-CVRP, where the term $\Theta(1/\sqrt{k})$ is larger than $1/3000$ for all $k\leq 10^7$.

\textbf{$k$-MCVRP}. Few results are available in the literature.
%Note that $\alpha\approx3/2$.
Based on a modification of ITP, Li and Simchi-Levi~\cite{tight} proposed a cycle-partition algorithm, which achieves a ratio of $2\alpha+1-\alpha/k$ for splittable and unit-demand $k$-MCVRP and a ratio of $2\alpha+2-2\alpha/k$ for unsplittable $k$-MCVRP. The only known improvement was made by Harks \emph{et al.}~\cite{HarksKM13}, where they proposed a tree-partition algorithm with an improved 4-approximation ratio for unsplittable $k$-MCVRP. Their algorithm also implies a 4.38-approximation ratio for a more general problem, called \emph{Capacitated Location Routing}, where we need to open some depots (with some cost) and then satisfy customers using vehicles in the opened depots.
When $k=\infty$, $k$-MCVRP becomes metric $m$-depot TSP.
For metric $m$-depot TSP, Rathinam~\emph{et al.}~\cite{rathinam2007resource} proposed a 2-approximation algorithm, and Xu~\emph{et al.}~\cite{xu2011analysis} proposed a $(2-1/m)$-approximation algorithm. Based on an edge exchange algorithm, Xu and Rodrigues~\cite{xu20153} obtained an improved $3/2$-approximation algorithm for any fixed $m$. Traub~\emph{et al.}~\cite{DBLP:journals/siamcomp/TraubVZ22} further improved the ratio to $\alpha+\varepsilon$ for any fixed $m$. Recently, Deppert~\emph{et al.}~\cite{deppert20233} obtained a randomized $(3/2+\varepsilon)$-approximation algorithm with a running time of $(1/\varepsilon)^{O(m\log m)}\cdot n^{O(1)}$. Notably, their algorithm even works with a \emph{variable} number of depots.

If $k$ is fixed, we will see that splittable $k$-MCVRP is equivalent to unit-demand $k$-MCVRP. Moreover, both unit-demand and unsplittable $k$-MCVRP can be reduced to the \emph{minimum weight $k$-set cover problem}.
%Unit-demand and unsplittable $k$-MCVRP can be reduced to the \emph{minimum weight $k$-set cover problem} if the capacity $k$ is fixed.
%We will see the splittable case is equivalent to the unit-demand case.
In minimum weight $k$-set cover, we are given a set of elements (called \emph{universe}), a set system with each set in it having a weight and at most $k$ elements, and we need to find a collection of sets in the set system with a minimum total weight that covers the universe. In the reduction, customers can be seen as the elements. There are at most $mn^{O(k)}$ feasible tours, and each tour can be seen as a set containing all customers in the tour with a weight of the tour. When $k$ is fixed, the reduction is polynomial.
A similar reduction for $k$-CVRP was mentioned in~\cite{zhaocvrp}.
%As we will see splittable $k$-MCVRP is equivalent to unit-demand $k$-MCVRP.
It is well-known~\cite{chvatal1979greedy} that the minimum weight $k$-set cover problem admits an approximation ratio of $H_k$, where $H_k\coloneqq 1+1/2+\cdots+1/k$ is the $k$-th harmonic number.
Hassin and Levin~\cite{hassin2005better} improved the ratio to $H_k-\Theta(1/k)$. Recently, using a non-obvious local search method, Gupta \emph{et al.}~\cite{gupta2023local} improved the ratio to $H_k-\Theta(\ln^2 k/k)$, which is better than 4 for any fixed $k\leq 30$. So, for some $k\leq 30$, the best ratios of $k$-MCVRP are $H_k-\Theta(\ln^2 k/k)$.

In our setting, vehicles must return to their starting depots, which is known as the \emph{fixed-destination} property~\cite{tight}. Li and Simchi-Levi~\cite{tight} also explored a non-fixed-destination version, where vehicles could end their routes at any depot. This version easily reduces to CVRP while maintaining the approximation ratio: one can regard all depots as a super-depot, and let the distance between a customer and the super-depot be the minimum weight of the edges between the customer and the depots.

Recently, Lai~\emph{et al.}~\cite{lai2023approximation} studied a variant of MCVRP, called \emph{$m$-Depot Split Delivery Vehicle Routing}, where the number of depots is still $m$, but the number of vehicles in each depot is limited and each vehicle can be used for at most one tour (one can also think that each depot contains only one vehicle, which can be used a limited number of times). When $m$ is fixed, they obtained a $(6-4/m)$-approximation algorithm. Carrasco Heine~\emph{et al.}~\cite{ejor/HeineDM23} considered a bifactor approximation algorithm for a variant of Capacitated Location Routing, where each depot has a capacity as well.

\subsection{Our Contributions}
Motivated by recent progress in approximating $k$-CVRP, we design improved approximation algorithms for $k$-MCVRP. For the sake of presentation, we assume $\alpha=3/2$.
The contributions are shown as follows.

Firstly, we review the cycle-partition algorithm in~\cite{tight} and then propose a refined tree-partition algorithm based on the idea in~\cite{HarksKM13}. Note that our refined algorithm has a better approximation ratio for fixed $k$. For splittable and unit-demand $k$-MCVRP, both of them are 4-approximation algorithms.
By making a trade-off between them and further using the result in~\cite{blauth2022improving}, we obtain an improved $(4-1/1500)$-approximation algorithm. It is worth noting that using solely the cycle-partition algorithm along with the result in~\cite{blauth2022improving} may yield only a $(4-1/3000)$-approximation ratio.

Secondly, using the LP-rounding method in~\cite{uncvrp}, we obtain an LP-based cycle-partition algorithm that can achieve a $(4+\ln2+\delta)$-approximation ratio for unsplittable $k$-MCVRP with any constant $\delta>0$.
By making a trade-off between the LP-based cycle-partition algorithm and the tree-partition algorithm and further using the result in~\cite{blauth2022improving}, we obtain an improved $(4-1/50000)$-approximation algorithm.

At last, we propose an LP-based tree-partition algorithm, which works for fixed $k$. Using the lower bounds of $k$-CVRP in~\cite{zhaocvrp}, we obtain an improved $(3+\ln2-\Theta(1/\sqrt{k}))$-approximation algorithm for all three versions of $k$-MCVRP, which is better than the current-best ratios for any $k>11$. Moreover, by making a trade-off between the LP-based tree-partition algorithm and the cycle-partition algorithm and further using the result in~\cite{blauth2022improving}, we show that the approximation ratio can be improved to $3+\ln2-$ $\max\{\Theta(1/\sqrt{k}),1/9000\}$ for splittable and unit-demand $k$-MCVRP.

An initial version of this paper was presented at the 29th International Conference on Computing and Combinatorics (COCOON 2023)~\cite{DBLP:conf/cocoon/ZhaoX23}.

\section{Preliminaries}\label{sec_pre}
\subsection{Definitions}
In MCVRP, we let $G=(V\cup D, E)$ denote the input complete graph, where vertices in $V$ represent customers and vertices in $D$ represent depots. There is a non-negative weight function $w: E\to \mathbb{R}_{\geq0}$ on the edges in $E$. We often write $w(u,v)$ to mean the weight of edge $uv$, instead of $w(uv)$. Note that $w(u,v)$ would be the same as the distance between $u$ and $v$.
The weight function $w$ is a semi-metric function, i.e., it is symmetric and satisfies the triangle inequality.
For any weight function $w': X\to \mathbb{R}_{\geq0}$, we extend it to the subsets of $X$, i.e., we define $w'(Y) = \sum_{x\in Y} w'(x)$ for any $Y\subseteq X$. There is a demand function $d$: $V\rightarrow\mathbb{N}_{\geq 1}$, where $d(v)$ is the demand required by customer $v\in V$. In the following, we let $\Delta=\sum_{v\in V}\min_{u\in D}d(v)w(u,v)$. For any component $S$, we use $v\in S$ (resp., $e\in S$) to denote a vertex (resp., an edge) of $S$, and let $w(S)\coloneqq\sum_{e\in S}w(e)$ and $d(S)\coloneqq\sum_{v\in S}d(v)$. For any graph $G'=(V',E')$ with a set of vertices $V''\subseteq V'$, we use $G'[V'']$ to denote the subgraph of $G'$ induced by $V''$.

A \emph{walk} in a graph is a succession of edges in the graph, where an edge may appear more than once. We will use a sequence of vertices to denote a walk. For example, $v_1v_2v_3\dots v_l$ means a walk with edges $v_1v_2$, $v_2v_3$, and so on.
A \emph{path} in a graph is a walk such that no vertex appears more than once in the sequence, and a \emph{cycle} is a walk such that only the first and the last vertices are the same.
A cycle containing $l$ edges is called an \emph{$l$-cycle} and the \emph{length} of it is $l$.
A \emph{spanning forest} in a graph is a forest that spans all vertices in the graph.
A \emph{constrained spanning forest} in graph $G$ is a spanning forest where each tree contains only one depot.

An \emph{itinerary} $I$ is a walk that starts and ends at the same depot and does not pass through any other depot. It is called an \emph{empty itinerary} and denote it by $I=\emptyset$ if there are no customer vertices on $I$, and a \emph{non-empty itinerary} otherwise.
A non-empty itinerary can be split into several minimal cycles containing only one depot, and each such cycle is called a \emph{tour}.
The Multidepot Capacitated Vehicle Routing Problem ($k$-MCVRP) can be described as follows.

\begin{definition}[$k$-MCVRP]
An instance $(G=(V\cup D,E),w,d,k)$ consists of:
\begin{itemize}
\item a complete graph $G$, where $V=\{v_1,\dots,v_n\}$ represents the $n$ customers and $D=\{u_1,\dots,u_m\}$ represents the $m$ depots,
\item a semi-metric weight function $w$: $(V\cup D)\times(V\cup D)\rightarrow\mathbb{R}_{\geq 0}$, which represents the distances,
\item a demand function $d$: $V\rightarrow\mathbb{N}_{\geq 1}$, where $d(v)$ is the demand required by customer $v\in V$,
\item the capacity $k\in\mathbb{Z}_{\geq 1}$ of vehicles that initially stays at each depot.
\end{itemize}
A feasible solution is a set of $m$ itineraries, with each containing one different depot:
\begin{itemize}
\item each tour delivers at most $k$ of the demand to customers on the tour,
\item the union of tours over all itineraries meets every customer's demand.
\end{itemize}
Specifically, the goal is to find a set of itineraries $\I=\{I_1,\dots,I_m\}$ where $I_i$ contains depot $u_i$, minimizing the total weight of the succession of edges in the walks in $\I$, i.e., $w(\I)\coloneqq \sum_{I\in \I}w(I)=\sum_{I\in \I}\sum_{e\in I}w(e)$.
\end{definition}

According to the property of the demand, we define three well-known versions. If each customer's demand must be delivered in one tour, we call it \emph{unsplittable $k$-MCVRP}. If a customer's demand can be split into several tours, we call it \emph{splittable $k$-MCVRP}. If each customer's demand is a unit, we call it \emph{unit-demand $k$-MCVRP}.

In the following, we use CVRP to denote MCVRP with $m=1$, i.e., only one depot.
Unless otherwise specified, $k$-MCVRP satisfies the fixed-destination property. If it holds the non-fixed-destination property, we called it \emph{non-fixed $k$-MCVRP}.

\subsection{Assumptions}
Note that in our problem the demand $d(v)$ may be very large since the capacity $k$ may be arbitrarily larger than $n$. For the sake of analysis, we make several assumptions that can be guaranteed by some simple observations or polynomial-time reductions.

\begin{assumption}~\label{reduction1}
For splittable and unsplittable $k$-MCVRP, each customer's demand is at most $k$.
\end{assumption}

Assumption~\ref{reduction1} can be guaranteed by polynomial-time reductions. Clearly, it holds for unsplittable $k$-MCVRP; otherwise, the problem has no feasible solution. For the splittable case, Dror and Trudeau~\cite{dror1990split} proved that two tours in $k$-CVRP sharing two common customers can be modified into two better tours, which clearly also holds for $k$-MCVRP. A \emph{trivial tour} is a tour that contains only one depot and only one customer. Consider a customer $v\in V$ and an optimal solution with no two tours sharing two common customers. Since there are $m$ depots and each tour contains one depot, there are at most $m(n-1)$ non-trivial tours containing $v$ that deliver at most $m(n-1)(k-1)$ demand to $v$. So, if $d(v)>m(n-1)(k-1)$, there exists an optimal solution that contains at least $l\coloneqq\ceil{\frac{d(v)-m(n-1)(k-1)}{k}}$ trivial tours. Note that there are $m$ kinds of trivial tours for $v$, and we can directly assign $l$ of the cheapest one. After doing these for all customers, every customer has a demand of at most $m(n-1)(k-1)$. Then, for each customer $v\in V$ such that $d(v)>k$, we can replace it with $O(mn)$ customers with zero inter-distance and each has a demand of at most $k$. The new instance is equivalent to the old one, and its size increases from $O(m+n)$ to $O(mn^2)$.

\begin{assumption}~\label{reduction2}
For splittable $k$-MCVRP with fixed $k$, each customer's demand is a unit.
\end{assumption}

By Assumption~\ref{reduction1}, each customer's demand is at most $k$. Under the case that $k$ is fixed, we can further replace each customer $v$ with $d(v)$ unit-demand customers with zero inter-distance. The size of the instance increases at most $O(nk)$. So, Assumption~\ref{reduction2} can be guaranteed by polynomial-time reductions. It states that splittable $k$-MCVRP with fixed $k$ can be reduced to unit-demand $k$-MCVRP in polynomial time while maintaining the approximation ratio.

\begin{assumption}~\label{observe1}
For unsplittable, splittable, and unit-demand $k$-MCVRP, there exists an optimal solution where each tour delivers an integer amount of demand to each customer in the tour.
\end{assumption}

Assumption~\ref{observe1} is a simple observation since in our problem both customers' demand and vehicles' capacity are integers. Note that it can be proved by a similar exchanging argument in~\cite{zhaocvrp}. So, in the following, we may only consider a tour that delivers an integer amount of demand to each customer in the tour.

After Assumption~\ref{observe1}, we know that for unit-demand $k$-MCVRP, there is an optimal solution consisting of a set of cycles, which intersect only at the depots.

\section{Lower Bounds}
To connect approximation algorithms for $k$-MCVRP with $k$-CVRP, it is convenient to consider non-fixed $k$-MCVRP. The first reason is that non-fixed $k$-MCVRP is a relaxation of $k$-MCVRP, and then an optimal solution of the former one provides a lower bound for the latter. Let $\OPT$ (resp., $\OPT'$) denote the weight of an optimal solution for $k$-MCVRP (resp., $k$-CVRP), and then we have $\OPT'\leq \OPT$. The second is that non-fixed $k$-MCVRP is equivalent to $k$-CVRP. The reduction implicitly used in \cite{tight} is shown as follows.

Given $G=(V\cup D, E)$, we obtain an undirected complete graph $H=(V\cup\{o\},F)$ by replacing the depots in $D$ with a new single depot, denoted by $o$. There is a weight function $c: F\to \mathbb{R}_{\geq0}$ on the edges in $F$, and it holds that $c(o,v)=\min_{u\in D}w(u,v)$ and $c(v,v')=\min\{c(o,v)+c(o,v'), w(v,v')\}$ for all $v,v'\in V$. It is easy to verify that the weight function $c$ is also a semi-metric function. Note that $\Delta=\sum_{v\in V}d(v)c(o,v)$.
Clearly, any feasible solution of $k$-CVRP in $H$ corresponds to a feasible solution for non-fixed $k$-MCVRP in $G$ with the same weight.
We call an edge $vv'\in E$ satisfying $w(v,v')>c(o,v)+c(o,v')$ a ``dummy'' edge. Any tour using a dummy edge $vv'$ can be transformed into two tours with a smaller weight by replacing $vv'$ with two edges $uv$ and $u'v'$ incident to depots such that $c(o,v)=w(u,v)$ and $c(o,v')=w(u',v')$.
So, any feasible solution of non-fixed $k$-MCVRP in $G$ can also be modified into a feasible solution for $k$-CVRP in $H$ with a non-increasing weight.

A \emph{Hamiltonian cycle} in a graph is a simple cycle that contains all vertices in the graph exactly once.
Let $C^*$ be a minimum cost Hamiltonian cycle in graph $H$.

We mention three lower bounds for $k$-CVRP, which also works for $k$-MCVRP.

\begin{lemma}[\cite{HaimovichK85}]\label{lb-tsp}
It holds that $\OPT\geq\OPT'\geq c(C^{*})$.
\end{lemma}
\begin{proof}
In an optimal solution of $k$-CVRP, it is easy to see that each vertex has an even degree. We can obtain an Euler tour and get a Hamiltonian cycle in $H$ by shortcutting. The Hamiltonian cycle has at least the cost of $C^*$. So, the lemma holds.
\end{proof}

\begin{lemma}[\cite{HaimovichK85}]\label{lb-delta}
It holds that $\OPT\geq\OPT'\geq(2/k)\Delta$.
\end{lemma}
\begin{proof}
Take a tour $C=ov_1\dots v_io$ in an optimal solution of $k$-CVRP. By the triangle inequality, we have $c(C)\geq 2c(o,v)$ for each $v\in C$. Since each tour has a capacity of at most $k$, we have $c(C)\geq(1/k)\sum_{v\in C}d(v)c(C)\geq (2/k)\sum_{v\in C}d(v)c(o,v)$. Each vertex $v\in V$ must be contained in a tour. So, $\OPT'\geq(2/k)\sum_{v\in V}d(v)c(o,v)=(2/k)\Delta$.
\end{proof}

Let $T^*$ denote an optimal spanning tree in $H$. Clearly, its cost is a lower bound of an optimal Hamiltonian cycle in $H$. By Lemma~\ref{lb-tsp}, we have the following lemma.

\begin{lemma}\label{lb-tree}
 It holds that $\OPT\geq\OPT'\geq c(T^*)$.
\end{lemma}

\section{Review of Previous Algorithms}
In this section, we review the cycle-partition algorithm in~\cite{tight} and the tree-partition algorithm in~\cite{HarksKM13} for $k$-MCVRP.

\subsection{The Cycle-Partition Algorithm}
The main idea of the cycle-partition algorithm~\cite{tight} is to construct a solution for non-fixed $k$-MCVRP based on the ITP or UITP algorithm for $k$-CVRP, and then modify the solution into a solution for $k$-MCVRP.

\textbf{The ITP and UITP algorithms.}
For splittable and unit-demand $k$-CVRP, the ITP algorithm operates by segmenting a given Hamiltonian cycle $C$ in graph $H$ into segments, each accommodating at most $k$ of demand. It then assigns two edges between the depot and the endpoints of each segment. The resulting solution has a weight of at most $(2/k)\Delta+c(C)$~\cite{HaimovichK85,altinkemer1990heuristics}.
For unsplittable $k$-CVRP, the UITP algorithm leverages the ITP algorithm with a capacity of $k/2$ to generate a solution for splittable and unit-demand $k$-CVRP. This solution is subsequently adjusted into a feasible solution for unsplittable $k$-CVRP. Altinkemer and Gavish \cite{altinkemer1987heuristics} proved that this modification incurs no additional cost.

%For splittable and unit-demand $k$-CVRP, given a Hamiltonian cycle $C$ in graph $H$, the ITP algorithm is to split the cycle into segments in a good way, with each containing at most $k$ of demand, and for each segment assign two edges between the new depot $o$ and endpoints of the segment. The solution has a weight of at most $(2/k)\Delta+c(C)$~\cite{HaimovichK85,altinkemer1990heuristics}.
%For unsplittable $k$-CVRP, the UITP algorithm is to use the ITP algorithm with a capacity of $k/2$ to obtain a solution for splittable and unit-demand $k$-CVRP, and then modify the solution into a feasible solution for unsplittable $k$-CVRP. Altinkemer and Gavish proved~\cite{altinkemer1987heuristics} that the modification does not take any additional cost.

%Note that we may assume w.l.o.g that $k$ is even; otherwise, we can double the capacity and the demand to get the same bound. A customer may be split into more than one tour, which is not allowed for unsplittable $k$-CVRP. Since each tour has a capacity of $k/2$, it can be easily modified into a feasible solution without using any additional cost~\cite{altinkemer1987heuristics}. We have the following lemma.

\begin{lemma}[\cite{HaimovichK85,altinkemer1990heuristics,altinkemer1987heuristics}]~\label{ITP}
Given a Hamiltonian cycle $C$ in $H$, for splittable and unit-demand $k$-CVRP, the ITP algorithm can use polynomial time to output a solution of cost at most $(2/k)\Delta+c(C)$; for unsplittable $k$-CVRP, the upper bound is $(4/k)\Delta+c(C)$.
\end{lemma}

\textbf{The cycle-partition algorithm.}
In $k$-MCVRP, the cycle-partition algorithm involves employing the ITP or UITP algorithm to derive a feasible solution for non-fixed $k$-MCVRP. This solution is then adjusted into a feasible solution for $k$-MCVRP, incurring some additional cost. Li and Simchi-Levi \cite{tight} proved that this additional cost is bounded by the cost of the Hamiltonian cycle utilized.
%For $k$-MCVRP, the cycle-partition algorithm is to use the ITP or UITP algorithm to obtain a feasible solution for non-fixed $k$-MCVRP, and then modify the solution into a feasible solution for $k$-MCVRP using some additional cost. Li and Simchi-Levi~\cite{tight} proved that the additional cost is exactly the cost of the Hamiltonian cycle used.

%For splittable and unit-demand cases, we get a solution with a weight of at most $(2/k)\Delta+w(C)$ for non-fixed $k$-MCVRP by Lemma~\ref{ITP}. It may contain a walk $u_sv_i\dots v_ju_t$ that starts with and ends at different depots $u_s$ and $u_t$ (the middle vertices are all customers). Let $w(v_i\dots v_j)$ denote the sum weight of edges in $v_i\dots v_j$. Note that $v_i\dots v_j$ is a segment obtained from splitting the Hamiltonian cycle $C$.
%Li and Simchi-Levi~\cite{tight} consider two possible modifications: $u_sv_i\dots v_ju_s$ and $u_tv_i\dots v_ju_t$, which corresponds to two feasible tours for $k$-MCVRP. By the triangle inequality, the better one has a weight of at most $w(u_s,v_i)+2w(v_i\dots v_j)+w(v_j,u_t)$. So, after the modification for all these walks, we get a feasible solution for $k$-MCVRP with a weight of at most $(2/k)\Delta+2w(C)$. Similarly, for the unsplittable case, the upper bound is $(4/k)\Delta+2w(C)$.

\begin{lemma}[\cite{tight}]~\label{cycle-partition}
Given a Hamiltonian cycle $C$ in $H$, for splittable and unit-demand $k$-MCVRP, there is a polynomial-time algorithm to output a solution of cost at most $(2/k)\Delta+2c(C)$; for unsplittable $k$-MCVRP, the upper bound is $(4/k)\Delta+2c(C)$.
\end{lemma}

Using the $3/2$-approximate Hamiltonian cycle~\cite{christofides1976worst,serdyukov1978some}, by Lemmas~\ref{lb-tsp} and \ref{lb-delta}, the cycle-partition algorithm achieves a 4-approximation ratio for splittable and unit-demand $k$-MCVRP, and a 5-approximation ratio for unsplittable $k$-MCVRP.

\subsection{The Tree-Partition Algorithm}
%The tree-partition algorithm relies on an optimal spanning tree within $H$ that corresponds to an optimal constrained spanning forest in $G$. This algorithm involves partitioning the constrained spanning forest into smaller components in a strategic manner. Each component is ensured to have a demand of at most $k$, with the additional constraint that any component devoid of depots contains a demand of at least $k/2$. Components containing one depot can be converted into tours by duplicating all edges within them and then shortening the resulting paths. For components lacking depots, the algorithm adds a single edge connecting them to the nearest depot with minimized weight. These components are subsequently transformed into tours using the same doubling and shortcutting method. Importantly, this algorithm is applicable to unsplittable, splittable, and unit-demand $k$-MCVRP simultaneously.
The tree-partition algorithm is based on an optimal spanning tree in graph $H$ that corresponds to an optimal constrained spanning forest in $G$. This algorithm aims to split the corresponding constrained spanning forest into edge-disjoint components ensuring that each component has a demand of at most $k$. Additionally, any component lacking depots must have a demand of at least $k/2$.
Components containing one depot can be converted into tours by doubling all edges within them and then shortcutting.
For components lacking depots, the algorithm adds a single edge connecting them to the nearest depot with minimized weight. These components are subsequently transformed into tours using the same doubling and shortcutting method.
Note that this algorithm is applicable to unsplittable, splittable, and unit-demand $k$-MCVRP simultaneously.

\begin{lemma}[\cite{HarksKM13}]~\label{tree-ITP}
For unsplittable, splittable, and unit-demand  $k$-MCVRP, there is a polynomial-time algorithm to output a solution of cost at most $(4/k)\Delta+2c(T^*)$.
\end{lemma}

By Lemmas~\ref{lb-delta} and \ref{lb-tree}, the tree-partition algorithm achieves a 4-approximation ratio for unsplittable, splittable, and unit-demand $k$-MCVRP.

\section{An Improvement for Splittable MCVRP}
In this section, we first propose a refined tree-partition algorithm based on the idea in~\cite{HarksKM13}. Our algorithm is simpler due to the previous assumptions. Moreover, our algorithm has a better approximation ratio for the case that $k$ is fixed. Then, based on recent progress in approximating $k$-CVRP~\cite{blauth2022improving}, we obtain an improved $(4-1/1500)$-approximation algorithm for splittable and unit-demand $k$-MCVRP.

\subsection{A Refined Tree-Partition Algorithm}
%In our algorithm, we start by assigning a single cheapest trivial tour for each customer $v\in V$ with demand $d(v) > \floor{k/2}$, as each customer's demand is at most $k$ according to Assumption~\ref{reduction1}. Let $V'$ denote the remaining customers. To satisfy customers in $V'$, we find an optimal spanning tree $T'^$ in the graph $H[V' \cup {o}]$. Note that $T'^$ corresponds to a constrained spanning forest in $G[V' \cup D]$, denoted by $\F$. We then consider a tree $T_u \in \F$ rooted at depot $u \in D$, and proceed to generate tours by splitting $T_u$, following the approach of Hark et al. in \cite{HarksKM13}. For each vertex $v \in T_u$, we denote the subtree rooted at $v$ and its set of children by $T_v$ and $Q_v$ respectively, and define $d(T_v) = \sum_{v' \in T_v} d(v')$.
In our algorithm, we start by assigning a single cheapest trivial tour for each customer $v\in V$ with demand $d(v)>\floor{k/2}$ since each customer's demand is at most $k$ according to Assumption~\ref{reduction1}.
Let $V'$ denote the remaining customers. To satisfy customers in $V'$, we find an optimal spanning tree $T'^*$ in $H[V'\cup\{o\}]$. Note that $T'^*$ corresponds to a constrained spanning forest in $G[V'\cup D]$, denoted by $\F$.  We then consider a tree $T_u\in\F$ rooted at depot $u\in D$, and proceed to generate tours by splitting $T_u$, following the approach of Hark \emph{et al.} in~\cite{HarksKM13}.
For each vertex $v\in T_u$, we denote the sub-tree rooted at $v$ and its set of children by $T_v$ and $Q_v$ respectively, and let $d(T_v)=\sum_{v'\in T_v}d(v')$.
Then, we consider the following two cases.
\begin{itemize}
\item If $d(T_u)\leq k$, it can be directly transformed into a tour by doubling and shortcutting.
\item Otherwise, we repeatedly perform the following until $d(T_u)\leq k$.\\
(1) Find a customer $v\in T_u$ such that $d(T_v)>k$ and $d(T_{v'})\leq k$ for every children $v'\in Q_v$, and consider the sub-trees $\T_v\coloneqq\{T_{v'}\mid v'\in Q_v\}$.\\
(2) Greedily partition them into $l$ sets $\T_1,\dots,\T_l$ such that $\floor{k/2}<d(\T_i)\leq k$ for each $i\in \{2,\dots,l\}$.\\
(3) For each such set, saying $\T_2$, combine them into a component $S$ by adding $v$ with edges joining $v$ and each tree in $\T_2$. Note that $S$ is a sub-tree of $T_v$.\\
(4) Find an edge $e_S$ with minimized weight connecting one depot in $D$ to one vertex in $S$.\\
(5) Obtain a tour satisfying all customers in trees of $\T_2$ by doubling $e_S$ with the edges in $S$ and shortcutting (note that we also need to shortcut $v$).\\
(6) After handling $\T_i$ for each $i\in {2,\dots,l}$, only $\T_1$ remains. If $d(\T_1)> \floor{k/2}$, handle it similarly to $\T_i$ with $i>1$. Otherwise, $d(\T_1)\leq \floor{k/2}$.\\
The remaining tree is denoted by $T'_u$, and the set of $v$'s children in $T'_u$ by $Q'_v$. Note that $d(\T_1)+d(v)\leq k$ since $d(v)\leq \floor{k/2}$. Hence, in $T'_u$, the condition $d(T'_v)>k$ and $d(T'_{v'})\leq k$ for every child $v'\in Q'_v$ will no longer hold, ensuring that the algorithm will terminate in polynomial time.
\end{itemize}

The algorithm is formally shown in Algorithm~\ref{algo:refined-tree-partition}.

\begin{algorithm}[t]
\caption{A refined tree-partition algorithm for $k$-MCVRP}
\label{algo:refined-tree-partition}
\small
\vspace*{2mm}
\textbf{Input:} Two undirected complete graphs: $G=(V\cup D, E)$ and $H=(V\cup\{o\}, F)$.\\
\textbf{Output:} A solution for $k$-MCVRP.

\begin{algorithmic}[1]
\State For each customer $v\in V$ with $d(v)>\floor{k/2}$, assign a trivial tour from $v$ to its nearest depot.\label{tour1}
\State Find an optimal spanning tree $T'^*$ in graph $H[V'\cup\{o\}]$.
\State Obtain the constrained spanning forest $\F$ in $G[V'\cup D]$ with respect to $T'^*$.
\For{each tree $T_u\in \F$}\Comment{$T_u$ is rooted at the depot $u$}
\While{$d(T_u)>k$}\label{loop}
\State Find $v\in T_u$ such that $d(T_v)>k$ and $d(T_{v'})\leq k$ for each $v'\in Q_v$.\Comment{$Q_v$ is the children set of $v$} \label{startloop}
%\State Find $v\in T_u$ such that $d(T_v)>k$ and $d(T_{v'})\leq k$ for each $v'\in Q_v$, where $Q_v$ is the set of children of $v$. \label{startloop}
\State Greedily partition the trees in $\T_v\coloneqq\{T_{v'}\mid v'\in Q_v\}$ into $l$ sets $\T_1,\dots,\T_l$ such that $\floor{k/2}<d(\T_i)\leq k$ for each $i\in \{2,\dots,l\}$.
\State Initialize $\mbox{Index}\coloneqq\{2,\dots,l\}$.
\If{$d(\T_1)>\floor{k/2}$}\label{sT_1}
\State $\mbox{Index}\coloneqq\mbox{Index}\cup\{1\}$.
\EndIf \label{tT_1}
\For{$i\in\mbox{Index}$}
\State Combine the trees in $\T_i$ into a component $S$ by adding $v$ with edges joining $v$ and each tree in $\T_i$.\label{start-tour2}
\State Find an edge $e_S$ with minimized weight connecting one depot in $D$ to one vertex in $S$.
\State Obtain a tour satisfying all customers in the trees of $\T_i$ by doubling $e_S$ with the edges in $S$ and shortcutting. \label{tour2}
\State Update $T_u$ by removing the component $S$ (except for) $v$ from $T_u$.\label{clean}
\EndFor
\EndWhile
\State Obtain a tour satisfying all customers in $T_u$ by doubling and shortcutting.\label{tour3}
\EndFor
\end{algorithmic}
\end{algorithm}

\begin{theorem}\label{tree-partition}
For unsplittable, splittable, and unit-demand $k$-MCVRP, Algorithm~\ref{algo:refined-tree-partition} can use polynomial time to output a solution of cost at most $\frac{2}{\floor{k/2}+1}\Delta+2c(T'^*)$.
\end{theorem}
\begin{proof}
First, we show that Algorithm~\ref{algo:refined-tree-partition} takes polynomial time. It is easy to see that all steps of Algorithm~\ref{algo:refined-tree-partition} take polynomial time except for the loop in Step~\ref{startloop}. Assuming that in the first loop, it finds a customer $v$ satisfying the condition in Step~\ref{startloop}.
As previously mentioned, by the end of this loop, we will have $d(T_v)\leq k$. So, the number of iterations in Step~\ref{loop} is bounded by the number of customers in $T_u$. Therefore, Algorithm~\ref{algo:refined-tree-partition} takes polynomial time.

Next, we show that the algorithm generates a feasible solution. Clearly, each tour in Steps~\ref{tour1}, \ref{tour2}, and \ref{tour3} has a capacity of at most $k$. Moreover, after assigning a tour in Step~\ref{tour2}, the algorithm removes all customers in the tour in Step~\ref{clean}. So, each customer is satisfied using only one tour.

At last, we analyze the total cost of all tours. We consider the following three cases.

\textbf{Case~1: tours in Step~\ref{tour1}.}
Let $\overline{V'}=V\setminus V'$.
Since we assign a trivial tour to each customer $v\in \overline{V'}$, the total cost of these tours is $\sum_{v\in \overline{V'}}2c(o,v)$. Since $d(v)>\floor{k/2}$ for every $v\in \overline{V'}$, the total cost is
\[
\sum_{v\in \overline{V'}}2c(o,v)\leq \frac{2}{\floor{k/2}+1}\sum_{v\in \overline{V'}}d(v)c(o,v).
\]

\textbf{Case~2: tours in Step~\ref{tour2}.} For each such tour, there exists a set of trees $\T_i$, a component $S$, and an edge $e_S$ corresponding to it. Since the tour is obtained by doubling $e_S$ with the edges in $S$ and shortcutting, its cost is bounded by twice the weight of edges in $S$ along with the weight of $e_S$. For the weight of $e_S$, we can get
\begin{align*}
w(e_S)=\min_{v'\in S}c(o,v')\leq \min_{v'\in \T_i} c(o,v')&\leq \sum_{v'\in\T_i}\frac{d(v')}{\sum_{v'\in\T_i}d(v')}c(o,v')\\
&\leq\frac{1}{\floor{k/2}+1}\sum_{v'\in\T_i}d(v')c(o,v'),
\end{align*}
where the first equality follows from the definition of $e_S$, and the last inequality follows from $d(\T_i)=\sum_{v'\in\T_i}d(v')>\floor{k/2}$.
Moreover, since each customer is contained in only one tour, the total cost of these tours on $e_S$ is
\[
\sum_{\T_i}2w(e_S)\leq\frac{2}{\floor{k/2}+1}\sum_{\T_i}\sum_{v'\in\T_i}d(v')c(o,v')\leq\frac{2}{\floor{k/2}+1}\sum_{v'\in V'}c(o,v').
\]
Moreover, the cost of these tours on the edges of $S$ is $\sum_{S}2w(S)$.

\textbf{Case~3: tours in Step~\ref{tour3}.} The cost of these tours is exactly twice the weight of the left edges after moving all components in Step~\ref{tour2} from each tree of $\F$. So, the total cost of these tours is $2w(\F)-\sum_{S}2w(S)$.

Therefore, the total cost of all tours is at most
\[
\frac{2}{\floor{k/2}+1}\sum_{v\in V}d(v)c(o,v)+2c(T'^*)=\frac{2}{\floor{k/2}+1}\Delta+2c(T'^*).
\]
Hence, the theorem holds.
\end{proof}

%Recall that the tree-partition algorithm in~\cite{HarksKM13} may get stuck in a loop.

\begin{lemma}\label{lb-tree+}
It holds that $c(C^*)\geq c(T'^*)$.
\end{lemma}
\begin{proof}
Let $C'^*$ denote an optimal Hamiltonian cycle in graph $H[V'\cup\{o\}]$. By the proof of Lemma~\ref{lb-tree}, we have $c(C'^*)\geq c(T'^*)$. Since we can obtain a Hamiltonian cycle in $H[V'\cup\{o\}]$ by shortcutting the optimal Hamiltonian cycle $C^*$ in $H$, by the triangle inequality, we have $c(C^*)\geq c(C'^*)$.
\end{proof}

By Theorem~\ref{tree-partition}, and Lemmas~\ref{lb-tsp} and~\ref{lb-tree+}, the refined tree-partition algorithm has an approximation ratio of $\frac{k}{\floor{k/2}+1}+2=4-\Theta(1/k)$, which is better than 4 for the case that $k$ is fixed. %Note that based on this algorithm we will obtain further improvements later for fixed $k$.
Next, we consider the improvement for general $k$.

\subsection{The Improvement}
Blauth~\emph{et al.}~\cite{blauth2022improving} made significant progress in approximating $k$-CVRP. We show that their method can be  extended to $k$-MCVRP to obtain a $(4-1/1500)$-approximation algorithm for splittable and unit-demand $k$-MCVRP.

Recall that $\OPT'$ is the cost of an optimal solution of $k$-CVRP in graph $H$.
%Their main result showed that the bound $\OPT'\geq(2/k)\Delta$ in Lemma~\ref{lb-delta} can be improved.
Fixing a constant $\varepsilon>0$, their main result demonstrates that if $(1-\varepsilon)\cdot\OPT' < (2/k)\Delta$, one can find a Hamiltonian cycle $C$ in polynomial time such that $\lim_{\varepsilon \rightarrow 0} c(C) = c(C^*)$.
%Fix a constant $\varepsilon>0$. Clearly, one can obtain an improvement using this better lower bound: $(1-\varepsilon)\cdot\OPT'\geq(2/k)\Delta$. Their main result showed that if $(1-\varepsilon)\cdot\OPT'<(2/k)\Delta$ one can find a Hamiltonian cycle $C$ such that $\lim_{\varepsilon\rightarrow 0}c(C)=c(C^*)$.

\begin{lemma}[\cite{blauth2022improving}]\label{good}
If $(1-\varepsilon)\cdot\OPT'<(2/k)\Delta$, there is a function $f: \mathbb{R}_{>0}\rightarrow\mathbb{R}_{>0}$ with $\lim_{\substack{\epsilon\rightarrow 0}}f(\epsilon)=0$ and a polynomial-time algorithm to get a Hamiltonian cycle $C$ in $H$ with $c(C)\leq (1+f(\epsilon)) \cdot\OPT'$.
\end{lemma}

The details of the function $f(\cdot)$ in Lemma~\ref{good} can be seen~in~\cite{blauth2022improving}.

\begin{theorem}\label{t1}
For splittable and unit-demand $k$-MCVRP, there is a polynomial-time $(4-1/1500)$-approximation algorithm.
\end{theorem}
\begin{proof}
Let $\varepsilon$ be a constant to be fixed later. We consider the following two cases.

\textbf{Case~1: $(1-\varepsilon)\cdot\OPT'\geq(2/k)\Delta$.} For this case, we call the refined tree-partition algorithm in Algorithm~\ref{algo:refined-tree-partition}.
By Theorem~\ref{tree-partition}, it generates a solution with a weight of at most $\frac{2}{\floor{k/2}+1}\Delta+2c(T'^*)\leq \frac{4}{k}\Delta+2c(T'^*)$. Since $\OPT\geq\OPT'$, by Lemmas~\ref{lb-tsp} and~\ref{lb-tree+}, it has an approximation ratio of $2(1-\varepsilon)+2=4-2\varepsilon$.

\textbf{Case~2: $(1-\varepsilon)\cdot\OPT'< (2/k)\Delta$.} For this case, we use the Hamiltonian $C$ in Lemma~\ref{good} to call the cycle-partition algorithm. By Lemma~\ref{cycle-partition}, it yields a solution with a weight of at most $(2/k)\Delta+2c(C)$. By Lemma~\ref{lb-delta}, it achieves an approximation ratio of $1+2(1+f(\varepsilon))=3+2f(\varepsilon)$.

Therefore, the above algorithm achieves an approximation ratio of
\[
\min_{\varepsilon>0}\max\{4-2\varepsilon,\ 3+2f(\varepsilon)\}.
\]
By calculation, we may set $\varepsilon=1/3000$. So, the approximation ratio is $4-1/1500$.
\end{proof}

Theorem~\ref{t1} is achieved by making a trade-off between the cycle-partition algorithm and the refined tree-partition algorithm. If we only use the cycle-partition algorithm, i.e., for the case that $(1-\varepsilon)\cdot\OPT'\geq(2/k)\Delta$, we also use it with a $3/2$-approximate Hamiltonian cycle, we may only get a $(4-1/3000)$-approximation algorithm.

\section{An Improvement for Unsplittable MCVRP}
In this section, we study unsplittable $k$-MCVRP. Friggstad~\emph{et al.}~\cite{uncvrp} proposed an improved LP-based approximation algorithm for unsplittable $k$-CVRP. We show that it can be used to get an LP-based cycle-partition algorithm for unsplittable $k$-MCVRP, achieving an improved $(4-1/50000)$-approximation ratio.

\subsection{An LP-Based Cycle-Partition Algorithm}
As mentioned in~\cite{zhaocvrp}, unsplittable $k$-CVRP can be reduced to the minimum weight $k$-set cover problem when $k$ is fixed.
In this reduction, customers are treated as elements, and each tour represents a set containing all customers in the tour with a weight equal to the tour's cost. But, when $k$ is not fixed, the number of feasible tours becomes $n^{O(k)}$, rendering the reduction non-polynomial.

The LP-based approximation algorithm proposed in~\cite{uncvrp} addresses this challenge by introducing a constant $0<\delta<1$ and constructing an LP only for customers $v$ with $d(v)\geq \delta k$. Consequently, the number of feasible tours reduces to $n^{O(1/\delta)}$, ensuring a polynomial bound. Using the well-known randomized LP-rounding technique (detailed in~\cite{williamson2011design}), a set of tours is obtained, forming a partial solution with an expected cost of $\ln2\cdot\OPT$. Then, tours are devised to satisfy the remaining customers based on a variant of the UITP algorithm, yielding an expected cost of $\frac{1}{1-\delta}\cdot(2/k)\Delta+c(C)$, where $C$ is a given Hamiltonian cycle in $H$. With $\delta$ chosen as a small constant, we have that
%The main idea of the LP-based approximation algorithm in~\cite{uncvrp} is that fixing a constant $0<\delta<1$ they build an LP (in the form of set cover) only for customers $v$ with $d(v)\geq \delta k$. Then, the number of feasible tours is $n^{O(1/\delta)}$ which is polynomially bounded. Using the well-known randomized LP-rounding method (see~\cite{williamson2011design}), they obtain a set of tours that forms a partial solution with a cost of $\ln2\cdot\OPT$. Then, they design tours to satisfy the left customers based on a variant of the UITP algorithm with an excepted cost of $\frac{1}{1-\delta}\cdot(2/k)\Delta+c(C)$, where $C$ is a given Hamiltonian cycle in graph $H$. Letting $\delta$ be a small constant, we have the following lemma.

\begin{lemma}[\cite{uncvrp}]\label{LP-UITP}
Given a Hamiltonian cycle $C$ in $H$, for unsplittable $k$-CVRP with any constant $\delta>0$, there is a polynomial-time algorithm to output a solution of cost at most $(\ln2+\delta)\cdot\OPT'+(2/k)\Delta+c(C)$.
\end{lemma}

For unsplittable $k$-MCVRP, we can use the same idea outlined above (details are put in Section~\ref{A.2}). By fixing a constant $0<\delta<1$, we build an LP only for customers $v$ with $d(v)\geq \delta k$. The partial solution derived from randomized LP-rounding achieves an expected weight of $\ln2\cdot\OPT$. For the remaining customers, we derive a solution for $k$-CVRP in $H$ with an expected cost of $\frac{1}{1-\delta}\cdot(2/k)\Delta+c(C)$. Given that this solution is based the idea of the UITP algorithm, we adjust it to form a set of feasible tours for $k$-MCVRP, incurring an additional cost of at most $c(H)$, akin to Lemma~\ref{cycle-partition}. %We can get the following theorem.
%For unsplittable $k$-MCVRP, we can use the same idea (details are put in Appendix~\ref{A.2}). Fixing a constant $0<\delta<1$, we build an LP (in the form of set cover) only for customers $v$ with $d(v)\geq \delta k$. A partial solution based on randomized LP-rounding has a weight of $\ln2\cdot\OPT$. For left customers, we obtain a solution for $k$-CVRP in $H$ with an excepted cost of $\frac{1}{1-\delta}\cdot(2/k)\Delta+c(C)$. Since the latter is based on the idea of the UITP algorithm, we can modify them into a set of feasible tours for $k$-MCVRP using an additional cost of $c(H)$ like Lemma~\ref{cycle-partition}. So, we can get the following theorem.

\begin{theorem}\label{lp-cycle-partition}
Given a Hamiltonian cycle $C$ in $H$, for unsplittable $k$-MCVRP with any constant $\delta>0$, the LP-based cycle-partition algorithm can use polynomial time to output a solution of cost at most $(\ln2+\delta)\cdot\OPT+(2/k)\Delta+2c(C)$.
\end{theorem}
\begin{proof}
The details are put in Section~\ref{A.2}.
\end{proof}

\subsection{The Improvement}
By using the LP-based cycle-partition algorithm along with a $3/2$-approximate Hamiltonian cycle, by Lemmas~\ref{lb-tsp} and \ref{lb-delta}, we can obtain a $(4+\ln2+\delta)$-approximation ratio for any constant $\delta>0$, which is even worse than the 4-approximation ratio obtained by the tree-partition algorithm.
However, the advantage of the LP-based cycle-partition algorithm lies in its ability to leverage the result in Lemma~\ref{good}. %Intuitively, if $(2/k)\Delta$ is close to $\OPT'$, we can find a close-to-optimal Hamiltonian cycle by Lemma~\ref{good}, and then the LP-based cycle-partition algorithm achieves a $(3+\ln2+\delta<4)$-approximation; otherwise, the tree-partition algorithm has a better than 4-approximation ratio.

\begin{theorem}\label{t2}
For unsplittable $k$-MCVRP, there is a polynomial-time $(4-1/50000)$-approximation algorithm.
\end{theorem}
\begin{proof}
Let $\delta$ and $\varepsilon$ be constants to be fixed later. We consider the following two cases.

\textbf{Case~1: $(1-\varepsilon)\cdot\OPT'\geq(2/k)\Delta$.} For this case, we call the refined tree-partition algorithm in Algorithm~\ref{algo:refined-tree-partition}. By the proof of Theorem~\ref{t1}, it achieves an approximation ratio of $2(1-\varepsilon)+2=4-2\varepsilon$.

\textbf{Case~2: $(1-\varepsilon)\cdot\OPT'< (2/k)\Delta$.} For this case, we use the Hamiltonian $C$ in Lemma~\ref{good} to call the LP-based cycle-partition algorithm. By Theorem~\ref{lp-cycle-partition}, it yields a solution with a weight of at most $(\ln2+\delta)\cdot\OPT'+(2/k)\Delta+2c(C)$. By Lemma~\ref{lb-delta}, it has an approximation ratio of $\ln2+1+2(1+f(\varepsilon))+\delta=3+\ln2+2f(\varepsilon)+\delta$.

Therefore, the above algorithm achieves an approximation ratio of
\[
\min_{\varepsilon>0}\max\{4-2\varepsilon,\ 3+\ln2+2f(\varepsilon)+\delta\}.
\]
With selection of $\delta$ to be sufficiently small, by calculation we may set $\varepsilon=1/100000$. So, the approximation ratio is $4-1/50000$.
%So, the combined algorithm has an approximation ratio of $\min_\varepsilon\max\{4-2\varepsilon,3+\ln2+2f(\varepsilon)+\delta\}$. By the calculation, the best choice for $\varepsilon$ is $1/100000$. So, the approximation ratio is $4-1/50000$.
\end{proof}

\section{An Improvement for $k$-MCVRP with Fixed Capacity}
In this section, we make further improvements for the case that the capacity $k$ is fixed. We propose an LP-based tree-partition algorithm based on the refined tree-partition algorithm with the LP-rounding method. This algorithm achieves an approximation ratio of $3+\ln2-\Theta(1/\sqrt{k})$. Then, by further combining the result in Lemma~\ref{good}, we also obtain a $(3+\ln2-1/9000)$-approximation algorithm for splittable and unit-demand $k$-MCVRP. Note that the former is better when $k$ is a fixed constant less than $3\times 10^8$.

\subsection{An LP-Based Tree-Partition Algorithm}
%The LP-Based Tree-Patr
After Assumption~\ref{observe1} we only consider a tour that delivers an integer amount of demand to each customer in the tour. Since $k$ is fixed, there are at most $mn^{O(k)}$ feasible tours for $k$-MCVRP. Moreover,  by Assumption~\ref{reduction2}, for splittable $k$-MCVRP each customer's demand is a unit. Denote the set of feasible tours by $\C$, and define a variable $x_C$ for each tour $C\in\C$. We build the following LP.
\begin{alignat}{3}
\text{minimize} & \quad & \sum_{C\in\C} w(C)\cdot x_C \notag\\
\text{subject to} & \quad & \sum_{\substack{C\in \C:\\ v\in C}}x_C \geq\   & 1, \quad && \forall\  v \in V, \notag\\
&& x_C \geq\   & 0, \quad && \forall\  C \in \C. \notag
\end{alignat}

The LP-based tree-partition algorithm is shown in Algorithm~\ref{algo:lp-tree-partition}. %Note that it has a similar framework with the LP-based cycle-partition algorithm.

\begin{algorithm}[t]
\caption{An LP-based tree-partition algorithm for $k$-MCVRP}
\label{algo:lp-tree-partition}
\small
\vspace*{2mm}
\textbf{Input:} Two undirected complete graphs: $G=(V\cup D, E)$ and $H=(V\cup\{o\}, F)$, and a constant $\gamma\geq 0$. \\
\textbf{Output:} A solution for $k$-MCVRP.

\begin{algorithmic}[1]
\State Solve the LP in polynomial time.
\For{$C\in \C$}
Put the tour $C$ into solution with a probability of $\min\{\gamma\cdot x_C, 1\}$.
\EndFor
\State For each customer appearing in multiple tours, we execute shortcutting on the tours to ensure that the customer appears in only one tour. \Comment{Due to the randomized rounding, some customers might appear in more than one tour.}
\label{cleanup}
\State Let $\widetilde{V}$ represent the set of customers that remain unsatisfied.
\State Obtain two new undirected complete graphs: $\widetilde{G}=G[\widetilde{V}\cup D]$ and $\widetilde{H}=G[\widetilde{V}\cup \{o\}]$.
\State Call the refined tree-partition algorithm in Algorithm~\ref{algo:refined-tree-partition}.
\end{algorithmic}
\end{algorithm}

Denote the set of tours obtained by the randomized LP-rounding method as $\C_1$ and the set of tours obtained by the refined tree-partition algorithm as $\C_2$. Next, we analyze the expected weight of $\C_1$ and $\C_2$. Next, we analyze the total expected weight of tours in $\C_1$ and $\C_2$, following a framework commonly used for the randomized LP-rounding method~\cite{williamson2011design,uncvrp}.

\begin{lemma}\label{lp1}
It holds that $\EE{w(\C_1)}\leq \gamma\cdot\OPT$.
\end{lemma}
\begin{proof}
Since the LP is a relaxation of $k$-MCVRP, we can get $\sum_{C\in\C}x_C\cdot w(C)\leq \OPT$. Each tour $C\in\C$ is chosen with a probability of at most $\gamma\cdot x_C$. So, the chosen tours have an expected weight of at most $\sum_{C\in\C}\gamma\cdot x_C\cdot w(C)\leq \gamma\cdot\OPT$. The shortcutting in Step~\ref{cleanup} does not increase the weight by the triangle inequality. Therefore, we have $\EE{w(\C_1)}\leq\sum_{C\in\C}\gamma\cdot x_C\cdot w(C)\leq \gamma\cdot\OPT$.
\end{proof}

\begin{lemma}\label{lp2}
It holds that $\PP{v\in\widetilde{V}}\leq e^{-\gamma}$ for every $v\in V$.
\end{lemma}
\begin{proof}
Let $\C_v$ be all feasible tours that contain $v$. Then, we have
\[
\PP{v\in\widetilde{V}}=\prod_{C\in\C_v}\PP{C~\mbox{is}~\mbox{not}~\mbox{chosen}}=\prod_{C\in\C_v}\lrA{1-\min\{\gamma\cdot x_C,\ 1\}}.
\]
We assume that $\min\{\gamma\cdot x_C,\ 1\}=\gamma\cdot x_C$ for any $C\in\C_v$; otherwise, $\PP{v\in\widetilde{V}}=0\leq e^{-\gamma}$, which holds trivially. Then, we have
\[
\PP{v\in\widetilde{V}}=\prod_{C\in\C_v}\PP{C~\mbox{is}~\mbox{not}~\mbox{chosen}}=\prod_{C\in\C_v}(1-\gamma\cdot x_C)\leq e^{-\sum_{C\in\C_v}\gamma\cdot x_C}\leq e^{-\gamma},
\]
where the first inequality follows from $1-x\leq e^{-x}$ for any $0\leq x\leq 1$, and the second from $\sum_{v\in\C_v}x_C\geq 1$ by the definition of the LP.
\end{proof}

Assumptions~\ref{reduction2} and \ref{observe1} also imply that there exists an optimal solution of $k$-CVRP in graph $H$ that consists of a set of simple cycles intersecting only at the depot.
By deleting the edge with the highest weight from each cycle, we can obtain a spanning tree in $H$. In the following, we denote this spanning tree by $T^{**}$.

\begin{lemma}\label{lp3}
It holds that $\EE{w(\C_2)}\leq e^{-\gamma}\cdot\frac{2}{\floor{k/2}+1}\Delta+2c(T^{**})$.
\end{lemma}
\begin{proof}
By Algorithms~\ref{algo:refined-tree-partition} and~\ref{algo:lp-tree-partition}, and Theorem~\ref{tree-partition}, we have
\[
w(\C_2)\leq \frac{2}{\floor{k/2}+1}\widetilde{\Delta}+2c(\widetilde{T'^*}),
\]
where $\widetilde{\Delta}\coloneqq\sum_{v\in\widetilde{V}}d(v)c(o,v)$, $\widetilde{V'}\coloneqq\{v\mid v\in \widetilde{V},\ d(v)>\floor{k/2}\}$, and $\widetilde{T'^*}$ is an optimal spanning tree in $H[\widetilde{V'}\cup\{o\}]$.

By Lemma~\ref{lp2}, we have
\[
\EE{\widetilde{\Delta}}=\sum_{v\in V}d(v)c(o,v)\cdot\PP{v\in\widetilde{V}}\leq e^{-\gamma}\cdot\sum_{v\in V}d(v)c(o,v)=e^{-\gamma}\cdot\Delta.
\]

Next, we analyze $c(\widetilde{T'^*})$. Note that $T^{**}$ is a spanning tree in $H$ such that it consists of a set of paths starting from $v$. Therefore, we can shortcut all customers in $V\setminus \widetilde{V'}$ to obtain another spanning tree in $H[\widetilde{V'}\cup\{o\}]$. Since $\widetilde{T'^*}$ is an optimal spanning tree in $H[\widetilde{V'}\cup\{o\}]$, we have $c(\widetilde{T'^*})\leq c(T^{**})$ by the triangle inequality.
\end{proof}

\begin{theorem}\label{lp-tree-partition}
For unsplittable, splittable, and unit-demand $k$-MCVRP with any constant $\gamma\geq 0$, the LP-based tree-partition algorithm can use polynomial time to output a solution with an expected cost of at most $\gamma\cdot\OPT+e^{-\gamma}\cdot\frac{2}{\floor{k/2}+1}\Delta+2c(T^{**})$.
\end{theorem}

The algorithm can be efficiently derandomized using the method of conditional expectations~\cite{williamson2011design,uncvrp}.

\subsection{The Analysis}
In this subsection, we show that with a well-chosen value of $\gamma$ the LP-based tree-partition algorithm achieves an approximation ratio of $3+\ln2-\Theta(1/\sqrt{k})$ for all the three versions of $k$-MCVRP.
After Assumption~\ref{reduction2}, splittable $k$-MCVRP is equivalent to unit-demand $k$-MCVRP.
%Moreover, since an optimal solution of unit-demand $k$-MCVRP is a lower bound of unsplittable $k$-MCVRP, the analyzed approximation ratio for unit-demand $k$-MCVRP also holds for unsplittable $k$-MCVRP. Therefore, we only analyze the approximation ratio for unit-demand $k$-MCVRP in the following.
Therefore, we only need to analyze the approximation ratios for the unit-demand and the unsplittable cases. We first analyze the unit-demand case.
%At last, we will show that we can also obtain the same result for unsplittable $k$-MCVRP.

We consider two lower bounds that were implicitly used in~\cite{zhaocvrp}.
\begin{lemma}[\cite{zhaocvrp}]\label{delta+tree}
For unit-demand $k$-CVRP in $H$, if there is a solution $\I$ consisting of a set of $(k+1)$-cycles that intersect only at the depot, and $T$ is a spanning tree in $H$ obtained by deleting the edge with the highest weight from each cycle in $\I$, it holds that
\[
\Delta\leq\lrA{\frac{k+2}{2}-\sum_{i=1}^{m'}ix_i}c(\I)\quad \mbox{and}\quad c(T)\leq\lrA{1-\frac{1}{2}\max_{1\leq i\leq m'}x_i}c(\I),
\]
where $m'\coloneqq\ceil{\frac{k+1}{2}}$, $x_i\geq 0$, and $\sum_{i=1}^{m'}x_i=1$.
\end{lemma}

Although the solution in Lemma~\ref{delta+tree} was considered as the optimal solution in \cite{zhaocvrp}, the analysis can be easily extended to any solution since it does not use the property of the optimality of the solution.
Note that $x_i$ in Lemma~\ref{delta+tree} measures some properties of the solution.

After Assumption~\ref{observe1}, the optimal solution of unit-demand $k$-CVRP in $H$ consists of a set of cycles intersecting only at the depot $o$. The cycles may not all be $(k+1)$-cycles. Recall the following result.

\begin{lemma}[\cite{zhaocvrp}]\label{add}
For unit-demand $k$-CVRP in $H$, if there is a solution $\I$ consisting of a set of cycles that intersect only at the depot $o$, after adding $k^2n+k-(n\bmod k)$ unit-demand customers on the depot, there exists a corresponding solution $\I'$ with the same cost consisting of a set of $(k+1)$-cycles only, and the cycles intersect only at $o$.
\end{lemma}
\begin{proof}
We include the proof for the sake of completeness.

Assume that $\I$ contains $m_i$ of $(i+1)$-cycles, where $1\leq i\leq k$. After adding $k^2n+k-(n\bmod k)$ unit-demand customers on the depot, each $(i+1)$-cycle can be transformed into a $(k+1)$-cycle with the same cost by using $k-i$ unit-demand customers on the depot. Hence, all cycles in $\I$ can be transformed into a set of $(k+1)$-cycles, denoted as $\I'_1$, using $\sum_{i=1}^{k}m_i(k-i)$ customers. These cycles intersect only at the depot, and it satisfies $c(\I'_1)=c(\I)$.

There are still $r\coloneqq k^2n+k-(n\bmod k)-\sum_{i=1}^{k}m_i(k-1)$ unit-demand customers on the depot not satisfied.
\begin{claim}
It holds that $k^2n+k-(n\bmod k)-\sum_{i=1}^{k}m_i(k-i)> 0$ and $k^2n+k-(n\bmod k)-\sum_{i=1}^{k}m_i(k-i)$ is divisible by $k$.
\end{claim}
\begin{proof}
Since $m_i\leq n$, we have $\sum_{i=1}^{k}m_i(k-i)\leq nk^2$. Hence,
\[
k^2n+k-(n\bmod k)-\sum_{i=1}^{k}m_i(k-i)\geq k-(n\bmod k)>0.
\]

To prove that $k^2n+k-(n\bmod k)-\sum_{i=1}^{k}m_i(k-i)$ is divisible by $k$, we only need to prove $n+\sum_{i=1}^{k}m_i(k-i)$ is divisible by $k$. Since $\I$ contains $m_i$ of $(i+1)$-cycles intersecting only at the depot, we have $n=\sum_{i=1}^{k}m_ii$, and then $n+\sum_{i=1}^{k}m_i(k-i)=\sum_{i=1}^{k}m_ik$, which is divisible by $k$.
\end{proof}

By the claim, the unsatisfied $r$ unit-demand customers on the depot can be satisfied using $r/k$ tours, denoted as $\I'_2$ with a zero cost, and each tour in $\I'_2$ is a $(k+1)$-cycle.

Therefore, there is a new solution $\I'\coloneqq\I'_1\cup\I'_2$ with respect to $\I$, and it satisfies that $c(\I')=c(\I)$.
\end{proof}

Before we use our approximation algorithm to generate a solution on an instance $I$, we may add a polynomial number of unit-demand customers on the depot to obtain a new instance $I'$. Then, if we obtain a $\rho$-approximate solution $\I'$ for $I'$, we can shortcut all the added customers to obtain a solution $\I$ with the same weight. Since the weight of the optimal solutions in $I$ and $I'$ are the same, $\I$ is still a $\rho$-approximate solution for $I$. Hence, adding a polynomial number of customers on the depot does not effect our approximation ratios.

Note that adding one customer on the depot $o$ in $H$ for $k$-CVRP is equivalent to adding one customer on any depot in $D$ of $G$ for $k$-MCVRP.
Since the optimal solution of unit-demand $k$-CVRP in $H$ consists of a set of cycles intersecting only at the depot $o$, by Lemma~\ref{add} and the previous analysis, we assume w.o.l.g. that for unit-demand $k$-CVRP in $H$ there exists an optimal solution consisting of a set of $(k+1)$-cycles, which intersect only at the depot.

Recall that $T^{**}$ is a spanning tree in $H$ obtained by deleting the edge with the highest weight from each cycle in the optimal solution. Since $\OPT'\leq\OPT$, by Lemma~\ref{delta+tree}, we have
\[
\Delta\leq\lrA{\frac{k+2}{2}-\sum_{i=1}^{m'}ix_i}\OPT\quad \mbox{and}\quad c(T^{**})\leq\lrA{1-\frac{1}{2}\max_{1\leq i\leq m'}x_i}\OPT,
\]
Moreover, by Theorem~\ref{lp-tree-partition}, for unit-demand $k$-MCVRP, the approximation ratio of the LP-based tree-partition algorithm is at most
\[
\max_{\substack{x_1,\dots, x_{m'}\geq 0\\ x_1+\dots+x_{m'}=1}}
\min_{\substack{\gamma\geq 0}}\lrC{\gamma+e^{-\gamma}\cdot\frac{2}{\floor{k/2}+1}\lrA{\frac{k+2}{2}-\sum_{i=1}^{m'}ix_i}+2\lrA{1-\frac{1}{2}\max_{1\leq i\leq m'}x_i}}.
\]
Before computing it, we show that it also holds for the unsplittable case.

\begin{lemma}
The approximation ratio of the LP-based tree-partition algorithm for unit-demand $k$-MCVRP also holds for the unsplittable case.
\end{lemma}
\begin{proof}
For unsplittable $k$-CVRP, the instance is represented by graph $H=(V\cup\{o\},F)$ with a cost function $c$ on $F$. We construct a new graph $\hat{H}=(\hat{V}\cup\{o\},\hat{F})$ by replacing each customer $v\in V$ with $d(v)$ unit-demand customers at the same location, which is regarded as an instance for the unit-demand case. With respect to $c$, there is a cost function $\hat{c}$ on $\hat{F}$. Recall that $\Delta=\sum_{v\in V}d(v)c(o,v)$. We let $\hat{\Delta}=\sum_{v\in \hat{V}}\hat{c}(o,v)$. It is easy to see that
\[
\Delta=\hat{\Delta}.
\]

Let $\I$ be the optimal solution for unsplittable $k$-CVRP in $H$, which corresponds to a solution with the same cost, denoted as $\hat{\I}$, for unit-demand $k$-CVRP in $\hat{H}$.

Note that $\hat{\I}$ contains a set of cycles intersecting only at the depot $o$.
The number of customers in $\hat{V}$ is $\size{\hat{V}}=\sum_{v\in V}d(v)$.
By Lemma~\ref{add}, for unit-demand $k$-CVRP in $\hat{H}$, if we add $k^2\cdot\size{\hat{V}}+k-(\size{\hat{V}}\bmod k)$ unit-demand customers on the depot $o$, there exists a solution with the same cost consisting of a set of $(k+1)$-cycles only, and the cycles intersect only at the depot. Hence, we add $k^2\cdot\size{\hat{V}}+k-(\size{\hat{V}}\bmod k)$ unit-demand customers on the depot $o$ in both $H$ and $\hat{H}$, which also implies that we add $k^2\cdot\size{\hat{V}}+k-(\size{\hat{V}}\bmod k)$ unit-demand customers on any depot in $D$ of $G$. Consequently, for unit-demand $k$-CVRP in $\hat{H}$, there exists a solution $\hat{\I}$ consisting of a set of $(k+1)$-cycles, which intersect only at the depot, and the cost of the solution $\hat{\I}$ is the same as that of the optimal solution $\I$ for unsplittable $k$-CVRP in $H$. Also note that the values $\Delta$ and $\hat{\Delta}$ keep unchanged.

Recall that $\OPT$ is the weight of an optimal solution for unsplittable $k$-MCVRP in $G$, $\OPT'$ is the weight of the optimal solution $\I$ for unsplittable $k$-CVRP in $H$, and $T^{**}$ is a spanning tree in $H$ obtained by deleting the edge with the highest weight from each cycle in the optimal solution $\I$.

Let $\hat{\OPT}$ denote the weight of the solution $\hat{\I}$ for unit-demand $k$-CVRP in $\hat{H}$, and $\hat{T}^{**}$ denote a spanning tree in $\hat{H}$ obtained by deleting the edge with the highest weight from each cycle in the solution $\hat{\I}$. By the proof of Lemma~\ref{add}, we can get
\[
\OPT\geq\OPT'=c(\I) = c(\hat{\I})=\hat{\OPT}\quad\mbox{and}\quad c(T^{**})=c(\hat{T}^{**}).
\]
Moreover, according to definitions of $\hat{\I}$ and $\hat{T}^{**}$, by Lemma~\ref{delta+tree}, we have
\[
\Delta=\hat{\Delta}\leq\lrA{\frac{k+2}{2}-\sum_{i=1}^{m'}ix_i}\OPT\quad \mbox{and}\quad c(T^{**})=c(\hat{T}^{**})\leq\lrA{1-\frac{1}{2}\max_{1\leq i\leq m'}x_i}\OPT,
\]
where $m'\coloneqq\ceil{\frac{k+1}{2}}$, $x_i\geq 0$, and $\sum_{i=1}^{m'}x_i=1$.

Therefore, the two lower bounds of $\OPT$ still hold for the unsplittable case. Then, the approximation ratio also holds for the unsplittable case.
\end{proof}

Computing this approximation ratio involves an optimization problem. We can even use computers to solve it in a brute-and-force way. Next, we give its precise value.

\begin{theorem}\label{t3}
For unsplittable, splittable, and unit-demand $k$-MCVRP with any fixed $k\in\mathbb{Z}_{\geq 3}$, by running the LP-based tree-partition algorithm with $\gamma=0$ and $\gamma=\gamma^*$ twice, we can obtain a solution in polynomial time with an approximation ratio of at most $\max\{g(\ceil{z_0}), g(\floor{z_0})\}$, where
\begin{itemize}
    \item $z_0\coloneqq\frac{\sqrt{4k+5}-1}{2}$,
    \item $g(z)\coloneqq3+\ln\left(\frac{k+1-z}{\floor{k/2}+1}\right)-\frac{1}{z}$,
    \item $\gamma^*\coloneqq\ln\left(\frac{k+1-\ceil{z_0}}{\floor{k/2}+1}\right)$ if $g(\ceil{z_0})\geq g(\floor{z_0})$ and $\gamma^*\coloneqq\ln\left(\frac{k+1-\floor{z_0}}{\floor{k/2}+1}\right)$ otherwise.
\end{itemize}
\end{theorem}
\begin{proof}
The proof is put in Section~\ref{A.3}.
\end{proof}

Our approximation ratio achieves $3+\ln2-\Theta(1/\sqrt{k})$, which performs well when $k$ is a small constant. However, it is worth noting that the $(H_k-\Theta(\ln^2 k/k))$-approximation ratio for $k$-set cover in \cite{gupta2023local} is better than ours for $k\leq 11$. %Therefore, for $k\leq 11$, the current-best ratios for $k$-MCVRP are $H_k-\Theta(\ln^2 k/k)$~\cite{gupta2023local}.

\subsection{A Further Improvement for Splittable $k$-MCVRP}
%Consider that $k$ is a big constant.
%If $(2/k)\Delta$ is close to $\OPT'$, for splittable and unit-demand $k$-MCVRP, the LP-based tree-partition achieves a $(3+\ln2)$-approximation ratio. Similarly, we can find an almost optimal Hamiltonian cycle by Lemma~\ref{good}, and then the cycle-partition algorithm achieves a $(3<3+\ln2)$-approximation ratio. So, we can obtain a further improvement for splittable and unit-demand $k$-MCVRP by making a trade-off between them.
We show that by making a trade-off between the LP-based tree-partition algorithm and the cycle-partition algorithm along with the result in~\cite{blauth2022improving}, we can further improve the approximation ratio to $3+\ln2-1/9000$ for splittable and unit-demand $k$-MCVRP.

\begin{theorem}\label{t4}
For splittable and unit-demand $k$-MCVRP, there is a polynomial-time $(3+\ln2-1/9000)$-approximation algorithm.
\end{theorem}
\begin{proof}
Let $\varepsilon$ be a constant to be fixed later. We consider the following two cases.

\textbf{Case~1: $(1-\varepsilon)\cdot\OPT'\geq(2/k)\Delta$.} For this case, we call the LP-based tree-partition algorithm with $\gamma=\ln2$. By Theorem~\ref{lp-tree-partition}, it generates a solution with a weight of at most $\ln2\cdot\OPT+\frac{1}{\floor{k/2}+1}\Delta+2c(T^{**})\leq\ln2\cdot\OPT+\frac{2}{k}\Delta+2c(T^{**})$. Since $c(T^{**})\leq \OPT'$ by Lemma~\ref{delta+tree} and $\OPT'\leq\OPT$, it achieves an approximation ratio of $\ln2+1-\varepsilon+2=3+\ln2-\varepsilon$.

\textbf{Case~2: $(1-\varepsilon)\cdot\OPT'< (2/k)\Delta$.} For this case, we use the Hamiltonian $C$ in Lemma~\ref{good} to call the cycle-partition algorithm. By the proof of Theorem~\ref{t1}, it has an approximation ratio of $3+2f(\varepsilon)$.

Therefore, the above algorithm achieves an approximation ratio of
\[
\min_{\varepsilon>0}\max\{3+\ln2-\varepsilon,\ 3+2f(\varepsilon)\}.
\]
By calculation, we may set $\varepsilon=1/9000$. The approximation ratio is $3+\ln2-1/9000$.
\end{proof}

By calculation, we observe that the result $3+\ln2-\Theta(1/\sqrt{k})$ in Theorem~\ref{t3} is better than $3+\ln2-1/9000$ in Theorem~\ref{t4} for any $k<3\times 10^8$. Note that for unsplittable $k$-MCVRP we cannot obtain further improvements using the same method. The reason is that even using an optimal Hamiltonian cycle the LP-based cycle-partition achieves an approximation ratio of only about $3+\ln2$ by Theorem~\ref{lp-cycle-partition}. So, there is no improvement compared with the LP-based tree-partition algorithm.

\section{The LP-Based Cycle-Partition Algorithm}\label{A.2}
In this section, we fix a constant $0<\delta<1$. For each customer $v\in V$, we say that it is \emph{big} if $d(v)\geq \delta k$, and \emph{small} otherwise.
Let $V_b\coloneqq\{v\mid d(v)\geq \delta k,v\in V\}$ be the set of big customers and $V_s\coloneqq V\setminus V_b$ be the set of small customers. Moreover, we define $\Delta_b\coloneqq\sum_{v\in V_b}d(v)c(o,v)$ and $\Delta_s\coloneqq\sum_{v\in V_s}d(v)c(o,v)$.

As mentioned, the LP-based cycle-partition algorithm mainly contains two phrases. The first is to obtain a partial solution with a weight of at most $\ln2\cdot\OPT$ based on the randomized LP-rounding for big customers; the second is to construct tours with an expected weight of $\frac{1}{1-\delta}(2/k)\Delta+2c(C)$ for customers that are still unsatisfied, where $C$ is a given Hamiltonian cycle in $H$. Note that the second phrase is based on a variant of the UITP algorithm for unsplittable $k$-CVRP, proposed in~\cite{uncvrp}.

\textbf{The first phrase.}
For unsplittable $k$-MCVRP, each feasible tour contains at most $1/\delta$ big customers. So, there are at most $mn^{O(1/\delta)}$ feasible tours containing only big customers. Denote the set of these tours by $\C$, and define a variable $x_C$ for each $C\in\C$. We have the following LP.
\begin{alignat}{3}
\text{minimize} & \quad & \sum_{C\in\C} w(C)\cdot x_C \notag\\
\text{subject to} & \quad & \sum_{\substack{C\in \C:\\ v\in C}}x_C \geq\   & 1, \quad && \forall\  v \in V, \notag\\
&& x_C \geq\   & 0, \quad && \forall\  C \in \C. \notag
\end{alignat}
Given an optimal solution of unsplittable $k$-MCVRP, one can shortcut all small customers to obtain a solution for big customers without increasing the weight. Since this LP considers all feasible tours on big customers, its value provides a lower bound for unsplittable $k$-MCVRP.
We obtain a set of tours, denoted as $\C_1$, by the randomized rounding: for each tour $C\in\C$, we select it into $\C_1$ with a probability of $\min\{\ln2\cdot x_C, 1\}$; For each customer appearing in multiple tours in $\C_1$, we execute shortcutting on the tours to ensure that the customer appears in only one tour.

Let $\widetilde{V}$ be the left customers that are still unsatisfied.
Friggstad~\emph{et al.}~\cite{uncvrp} proposed a variant of the UITP algorithm for unsplittable $k$-CVRP. In the following, we call it the \emph{$\delta$-UITP} algorithm. The idea of the $\delta$-UITP algorithm is still to split a Hamiltonian cycle in $H$ into segments in a good way, with each containing at most $k$ of demand, and for each segment assign two edges between the depot $o$ and endpoints of the segment. %As mentioned in~\cite{uncvrp}, one can use $O(n^2)$ time to find an optimal partition of the Hamiltonian cycle and then transform each segment into a feasible tour.
We have the following lemma.

\begin{lemma}[\cite{uncvrp}]\label{UITP+}
Given a Hamiltonian cycle $C$ in $H$, for unsplittable $k$-CVRP with any constant $\delta>0$, the $\delta$-UITP algorithm can use polynomial time to output a solution of cost at most $\frac{1}{1-\delta}(2/k)\Delta_s+\frac{1}{1-\delta}(4/k)\Delta_b+c(C)$.
\end{lemma}

\textbf{The second phrase.} Given a Hamiltonian cycle $C$ in $H$, we can get a Hamiltonian cycle $\widetilde{C}$ in $H[\widetilde{V}\cup\{o\}]$ by shortcutting. To satisfy customers in $\widetilde{V}$, we construct a set of tours, denoted as $\C_2$, by calling the $\delta$-UITP algorithm along with the Hamiltonian cycle $\widetilde{C}$ for unsplittable $k$-CVRP in $H[\widetilde{V}\cup\{o\}]$.
Then, we modify the tours in $\C_2$ into a set of feasible tours, denoted as $\C'_2$, for unsplittable $k$-MCVRP.

The modification works as follows.
With respect to the tours in $\C_2$ in $H$, we otain a set of minimal walks, denoted as $\W$, in $G$, where each walk starts at a depot and ends at a depot in $D$.
$\W$ may contain a walk $u_sv_i\dots v_ju_t$ that starts at one depot $u_s$ and ends at a different depot $u_t$ (the middle vertices are all customers). However, this is not feasible since for $k$-MCVRP each vehicle must return to the depot from which it started.
Let $w(v_i\dots v_j)$ denote the sum of the weights of edges in the segment $v_i\dots v_j$, which is derived from splitting the Hamiltonian cycle $\widetilde{C}$ by the $\delta$-UITP algorithm.
We consider two possible modifications: $u_sv_i\dots v_ju_s$ and $u_tv_i\dots v_ju_t$, which correspond to two feasible tours. By the triangle inequality, the better one has a weight of at most
\[
w(u_s,v_i)+2w(v_i\dots v_j)+w(v_j,u_t).
\]
This modification introduces an additional weight for the segment $v_i\dots v_j$. Since the segments obtained from splitting a Hamiltonian cycle using the $\delta$-UITP algorithm are edge-disjoint~\cite{uncvrp}, after modifying all these walks, we obtain a set of feasible tours $\C'_2$ for $k$-MCVRP incurring an additional cost of at most $c(\widetilde{C})$.

The LP-based cycle-partition algorithm is formally shown in Algorithm~\ref{algo:lp-cycle-partition}.

\begin{algorithm}[t]
\caption{An LP-based cycle-partition algorithm for unsplittable $k$-MCVRP}
\label{algo:lp-cycle-partition}
\small
\vspace*{2mm}
\textbf{Input:} Two undirected complete graphs: $G=(V\cup D, E)$ and $H=(V\cup\{o\}, F)$, a Hamiltonian cycle $C$ in graph $H$, and a constant $\delta>0$. \\
\textbf{Output:} A solution for unsplittable $k$-MCVRP.

\begin{algorithmic}[1]
\State Solve the LP in $m^{O(1)}n^{O(1/\delta)}$ time.
\For{$C'\in \C$}
Put the tour $C'$ into solution with a probability of $\min\{\ln2\cdot x_{C'}, 1\}$.
\EndFor
\State For each customer appearing in multiple tours, we execute shortcutting on the tours to ensure that the customer appears in only one tour. \Comment{Due to the randomized rounding, some customers might appear in more than one tour.}
\State Let $\widetilde{V}$ represent the set of customers that remain unsatisfied.
\State Obtain two new undirected complete graphs: $\widetilde{G}=G[\widetilde{V}\cup D]$ and $\widetilde{H}=G[\widetilde{V}\cup \{o\}]$, and a Hamiltonian cycle $\widetilde{C}$ in $\widetilde{H}$ by shortcutting $C$.
\State Call the $\delta$-UITP algorithm in Lemma~\ref{UITP+} along with $\widetilde{C}$ to obtain a set of tours $\C_2$ in $\widetilde{H}$ for unsplittable $k$-CVRP.
\State Obtain a set of minimal walks $\W$ that starts at a depot and ends at a depot in $\widetilde{G}$ with respect to the tours in $\C_2$.
\While{$W=u_sv_i\dots v_ju_t\in\W$ such that $u_s$ and $u_t$ are different depots}
\State Modify it to the better tour from $u_sv_i\dots v_ju_s$ and $u_tv_i\dots v_ju_t$.
\EndWhile
\State Obtain a set of feasible tours $\C'_2$ in $G$ with respect to $\W$, and put them into solution.
\end{algorithmic}
\end{algorithm}

The total cost of the LP-based cycle-partition algorithm is $w(\C_1)+w(\C'_2)$. Next, we analyze the expected weight of $\C_1$ and $\C'_2$, following a framework commonly used for the randomized LP-rounding method~\cite{williamson2011design,uncvrp}.

\begin{lemma}\label{zlp1}
It holds that $\EE{w(\C_1)}\leq \ln2\cdot\OPT$.
\end{lemma}
\begin{proof}
Since LP provides a lower bound of unsplittable $k$-MCVRP, $\sum_{C\in\C}x_C\cdot w(C)\leq \OPT$. Each tour $C\in\C$ is chosen with a probability of at most $\ln2\cdot x_C$. So, the chosen tours have an expected weight of at most $\sum_{C\in\C}\ln2\cdot x_C\cdot w(C)\leq\ln2\cdot\OPT$. The shortcutting in Step~\ref{cleanup} will not increase the weight by the triangle inequality. So, we have $\EE{w(\C_1)}\leq\ln2\cdot\OPT$.
\end{proof}

\begin{lemma}\label{zlp2}
It holds that $\PP{v\in\widetilde{V}\mid v\in V_s}=1$ and $\PP{v\in\widetilde{V}\mid v\in V_b}\leq 1/2$.
\end{lemma}
\begin{proof}
Clearly, small customers are still not satisfied since tours in LP are only for big customers. Hence, $\PP{v\in\widetilde{V}\mid v\in V_s}=1$. For a big customer $v\in V_b$, let $\C_v$ be the set of tours in $\C$ that contain $v$.
Then, we have
\[
\PP{v\in\widetilde{V}\mid v\in V_b}=\prod_{C\in\C_v}\PP{C~\mbox{is}~\mbox{not}~\mbox{chosen}}=\prod_{C\in\C_v}\lrA{1-\min\{\ln2\cdot x_C,\ 1\}}.
\]
We assume that $\min\{\ln2\cdot x_C,\ 1\}=\ln2\cdot x_C$ for any $C\in\C_v$; otherwise, $\PP{v\in\widetilde{V}\mid v\in V_b}=0\leq 1/2$, which holds trivially. Then, we have
\[
\PP{v\in\widetilde{V}\mid v\in V_b}=\prod_{C\in\C_v}(1-\ln2\cdot x_C)\leq e^{-\sum_{C\in\C_v}\ln2\cdot x_C}\leq e^{-\ln2}=\frac{1}{2},
\]
where the first inequality follows from $1-x\leq e^{-x}$ for any $0\leq x\leq 1$, and the second from $\sum_{v\in\C_v}x_C\geq 1$ by the definition of the LP.
\end{proof}

\begin{lemma}\label{zlp3}
It holds that $\EE{w(\C'_2)}\leq \frac{1}{1-\delta}(2/k)\Delta+2c(C)$.
\end{lemma}
\begin{proof}
Recall that $\widetilde{C}$ is obtained by shortcutting $C$ and tours in $\C'_2$ are obtained by modifying tours in $\C_2$ using an additional cost of $c(\widetilde{C})$. We have $w(\C'_2)\leq c(\C_2)+c(\widetilde{C})\leq c(\C_2)+c(C)$ by the triangle inequality. Since $\C_2$ is a set of tours obtained by calling the $\delta$-UITP algorithm along with the Hamiltonian cycle $\widetilde{C}$ for unsplittable $k$-CVRP in $H[\widetilde{V}\cup\{o\}]$, by Lemmas~\ref{UITP+} and \ref{zlp2}, we have
\begin{align*}
\EE{c(\C_2)}&\leq \frac{1}{1-\delta}(2/k)\sum_{v\in V_s}d(v)c(o,v)\cdot\PP{v\in\widetilde{V}\mid v\in V_s}\\
&\quad+\frac{1}{1-\delta}(4/k)\sum_{v\in V_b}d(v)c(o,v)\cdot\PP{v\in\widetilde{V}\mid v\in V_b}+c(C)\\
&\leq \frac{1}{1-\delta}(2/k)\Delta_s+\frac{1}{1-\delta}(2/k)\Delta_b+c(C)\\
&=\frac{1}{1-\delta}(2/k)\Delta+c(C),
\end{align*}
where the last equality follows from $\Delta_s+\Delta_b=\Delta$ by definition. Therefore, we can get $\EE{w(\C'_2)}\leq \frac{1}{1-\delta}(2/k)\Delta+2c(C)$.
\end{proof}

\begingroup
\def\thetheorem{\ref{lp-cycle-partition}}
\begin{theorem}
Given a Hamiltonian cycle $C$ in $H$, for unsplittable $k$-MCVRP with any constant $\delta>0$, the LP-based cycle-partition algorithm can use polynomial time to output a solution of cost at most $(\ln2+\delta)\cdot\OPT+(2/k)\Delta+2c(C)$.
\end{theorem}
\begin{proof}
By Lemmas~\ref{zlp1} and \ref{zlp3}, for unsplittable $k$-MCVRP with any constant $0<\delta'<1$, the LP-based cycle-partition algorithm has an expected weight of at most
\begin{align*}
&\ln2\cdot\OPT+\frac{1}{1-\delta'}(2/k)\Delta+2c(C)\\
&=\ln2\cdot\OPT+\frac{\delta'}{1-\delta'}(2/k)\Delta+(2/k)\Delta+2c(C),\\
&\leq \lrA{\ln2+\frac{\delta'}{1-\delta'}}\cdot\OPT+(2/k)\Delta+2c(C),
\end{align*}
where the inequality follows from $(2/k)\Delta\leq\OPT$ by Lemma~\ref{lb-delta}.

For any constant $\delta>0$, when setting $\delta'\leq \frac{\delta}{\delta+1}$, the expected weight is bounded by
\[
\lra{\ln2+\delta}\cdot\OPT+(2/k)\Delta+2c(C)
\]

Since the algorithm can be efficiently derandomized using the method of conditional expectations~\cite{williamson2011design,uncvrp}, Theorem~\ref{lp-cycle-partition} holds.
\end{proof}
\endgroup

\section{Proof of Theorem~\ref{t3}}~\label{A.3}
\begingroup
\def\thetheorem{\ref{t3}}
\begin{theorem}
For unsplittable, splittable, and unit-demand $k$-MCVRP with any fixed $k\in\mathbb{Z}_{\geq 3}$, by running the LP-based tree-partition algorithm with $\gamma=0$ and $\gamma=\gamma^*$ twice, we can obtain a solution in polynomial time with an approximation ratio of at most $\max\{g(\ceil{z_0}), g(\floor{z_0})\}$, where
\begin{itemize}
    \item $z_0\coloneqq\frac{\sqrt{4k+5}-1}{2}$,
    \item $g(z)\coloneqq3+\ln\left(\frac{k+1-z}{\floor{k/2}+1}\right)-\frac{1}{z}$,
    \item $\gamma^*\coloneqq\ln\left(\frac{k+1-\ceil{z_0}}{\floor{k/2}+1}\right)$ if $g(\ceil{z_0})\geq g(\floor{z_0})$ and $\gamma^*\coloneqq\ln\left(\frac{k+1-\floor{z_0}}{\floor{k/2}+1}\right)$ otherwise.
\end{itemize}
\end{theorem}
\begin{proof}
Recall that with the best choice of $\gamma$ the approximation ratio of the LP-based tree-partition algorithm is at most
\[
\min_{\substack{\gamma\geq 0}}\max_{\substack{x_1,\dots, x_{m'}\geq 0\\ x_1+\dots+x_{m'}=1}}\lrC{\gamma+e^{-\gamma}\cdot\frac{2}{\floor{k/2}+1}\lrA{\frac{k+2}{2}-\sum_{i=1}^{m'}ix_i}+2\lrA{1-\frac{1}{2}\max_{1\leq i\leq m'}x_i}}.
\]
Although $m'$ is required to be $\ceil{\frac{k+1}{2}}$, for the sake of analysis, we relax this constraint so that $m'=\infty$. Next, we explore some properties of $x_i$ in the worst case. %Note that the relaxation may only worsen our

Firstly, it is easy to observe that in the worst case we have $$x_1\geq x_2\geq\cdots x_{m'}$$ since if $x_p<x_q$ for some $p<q$ we can swap theirs values to get a larger approximation ratio. Hence, we have $$\max_{1\leq i\leq m'}x_i=x_1.$$
Moreover, if $x_1$ is fixed, in the worst case $\frac{k+2}{2}-\sum_{i=1}^{m'}ix_i$ should be maximized under the constants $x_1\geq x_2\geq\cdots\geq x_{m'}$ and $\sum_{i=1}^{m'}x_i=1$.
To maximize $\frac{k+2}{2}-\sum_{i=1}^{m'}ix_i$, we can greedily set $x_2$,..., $x_{m'}$, where for each $i$ with $2\leq i\leq m$, we set $x_i=x_{i-1}$ if $1-x_1-\cdots- x_{i-1}\geq x_{i-1}$, and $x_i=1-x_1-\cdots -x_{i-1}$ otherwise.
Hence, we can get
\[
x_1=x_2=x_3=\cdots =x_l>x_{l+1}\quad \mbox{and}\quad x_{l+2}=x_{l+3}=\cdots=x_{m'}=0,
\]
where $l=\floor{\frac{1}{x_1}}$ and $x_{l+1}=1-lx_1$.
Then, we have
\[
\sum_{i=1}^{m'}ix_i=\sum_{i=1}^{l}ix_1+(l+1)x_{l+1}=\frac{l(l+1)x_1}{2}+(l+1)(1-lx_1)=\frac{-l^2x_1-lx_1+2l+2}{2}.
\]
Therefore, the approximation ratio is given by
\[
\min_{\substack{\gamma\geq 0}}\max_{0<x_1\leq 1}\lrC{\gamma+e^{-\gamma}\cdot\frac{k-2l+l^2x_1+lx_1}{\floor{k/2}+1}+2-x_1}.
\]

Let
\[
f(\gamma, x)\coloneqq\gamma+e^{-\gamma}\cdot\frac{k-2l+l^2x+lx}{\floor{k/2}+1}+2-x,\quad \mbox{where}\quad l=\Floor{\frac{1}{x}}.
\]
The approximation ratio is given by $\min_{\substack{\gamma\geq 0}}\max_{0<x\leq 1}f(\gamma, x)$. Then, we explore some properties of $f(\gamma, x)$.

For any $z\in\mathbb{Z}_{\geq1}$, it is easy to verify
\[
\lim_{x\rightarrow\lrA{\frac{1}{z}}^-}f(\gamma, x)=\lim_{x\rightarrow\lrA{\frac{1}{z}}^+}f(\gamma, x).
\]
Moreover, under $x\in[\frac{1}{z+1},\frac{1}{z}]$, we can get
\[
k-z\leq\frac{k-2l+l^2x+lx}{\floor{k/2}+1}\leq k-z+1.
\]
Since $k\in\mathbb{Z}_{\geq1}$, we known that
\[
\frac{k-2l+l^2x+lx}{\floor{k/2}+1}\leq1\quad\Longleftrightarrow\quad x\leq \frac{1}{k-\floor{k/2}}.
\]
The approximation ratio is given by
\[
\max\lrC{\min_{\substack{\gamma\geq 0}}\max_{0<x\leq \frac{1}{k-\floor{k/2}}}f(\gamma, x),\quad\min_{\substack{\gamma\geq 0}}\max_{\frac{1}{k-\floor{k/2}}<x\leq 1}f(\gamma, x)}
\]

\textbf{Case~1: $0<x\leq \frac{1}{k-\floor{k/2}}$.} We consider the LP-based tree-partition algorithm with $\gamma=0$. We can get
\[
\max_{0<x\leq \frac{1}{k-\floor{k/2}}}f(0, x)=\frac{k-2l+l^2x+lx}{\floor{k/2}+1}+2-x.
\]
Since $0<x\leq \frac{1}{k-\floor{k/2}}$, we can get $l\geq k-\floor{k/2}$. For any $x\in(\frac{1}{z+1}, \frac{1}{z}]$ with $z\in\mathbb{Z}_{\geq k-\floor{k/2}}$, since $k\in\mathbb{Z}_{\geq 3}$, it is easy to verify that
\[
\frac{\partial f}{\partial x}=\frac{l^2+l}{\floor{k/2}+1}-1\geq 0.
\]
Hence, we have
\begin{align*}
\max_{0<x\leq \frac{1}{k-\floor{k/2}}}f(0, x)&\leq f\lrA{0,\frac{1}{k-\floor{k/2}}}\\
&=\frac{k-(k-\floor{k/2})+1}{\floor{k/2}+1}+2-\frac{1}{k-\floor{k/2}}\\
&=3-\frac{1}{k-\floor{k/2}}.
\end{align*}
The approximation ratio is bounded by $3-\frac{1}{k-\floor{k/2}}$.

\textbf{Case~2: $\frac{1}{k-\floor{k/2}}\leq x\leq 1$.} Note that
\[
\left. \frac{\partial f}{\partial \gamma} \right|_{\gamma=\ln\lrA{\frac{k-2l+l^2x+lx}{\floor{k/2}+1}}}=0.
\]
Let
\[
h(x)\coloneqq f\lrA{\ln\lrA{\frac{k-2l+l^2x+lx}{\floor{k/2}+1}}, x}=3+\ln\lrA{\frac{k-2l+l^2x+lx}{\floor{k/2}+1}}-x.
\]
Hence, we can get
\begin{align*}
\min_{\gamma\geq0}\max_{\frac{1}{k-\floor{k/2}}\leq x\leq 1}f(\gamma, x)&\leq \max_{x\leq 1}h(x).
\end{align*}
For any $x\in(\frac{1}{z+1}, \frac{1}{z})$ with $z\in\mathbb{Z}_{\geq 1}$, we have $l=\floor{\frac{1}{x}}=z$, and then we can get
\[
h'(x)=\frac{z^2+z}{k-2z+z^2x+zx}-1=\frac{(z^2+z)(1-x)+2z-k}{(z^2+z)x-2z+k}
\]
Note that
\[
(z^2+z)(1-x)+2z-k\in(z^2+2z-k-1,z^2+2z-k)~\mbox{and}~(z^2+z)x-2z+k\in(k-z,k-z+1).
\]
Hence, it holds $h'(x)>0$ or $h'(x)<0$ for any $x\in(\frac{1}{z+1}, \frac{1}{z})$. It is also easy to verify that
\[
\lim_{x\rightarrow\lrA{\frac{1}{z}}^-}h(x)=\lim_{x\rightarrow\lrA{\frac{1}{z}}^+}h(x).
\]
Then, we know that $h(x)$ attains its maximum value only if $1/x$ is an integer.

Let
\[
g(z)\coloneqq 3+\ln\left(\frac{k-2z+z^2\cdot\frac{1}{z}+z\cdot\frac{1}{z}}{\floor{k/2}+1}\right)-\frac{1}{z}=3+\ln\left(\frac{k-z+1}{\floor{k/2}+1}\right)-\frac{1}{z},\quad z\geq 1.
\]
For any $z\in\mathbb{Z}_{\geq 1}$, we have $g(z)=h(\frac{1}{z})$. Hence, the approximation ratio is given by
\[
\max_{z\in\mathbb{Z}_{\geq 1}}g(z).
\]
We have
\[
g'(z)=\frac{-1}{k-z+1}+\frac{1}{z^2}.
\]
Let $z_0\coloneqq\frac{\sqrt{4k+5}-1}{2}$.
When $z=z_0$, we have $g'(z)=0$. Since $z_0>1$, we have
\[
\max_{z\in\mathbb{Z}_{\geq1}} g(z)=\max\{g(\ceil{z_0}), g(\floor{z_0})\}.
\]
Let $z^*\coloneqq\arg\max_{z\in\mathbb{Z}_{\geq1}} g(z)$, and $x^*\coloneqq\frac{1}{z^*}$. Note that we have $z^*=\ceil{z_0}$ if $g(\ceil{z_0})\geq g(\floor{z_0})$, and $z^*=\floor{z_0}$ otherwise. Since $\floor{\frac{1}{x^*}}=z^*$, we know that
\begin{align*}
\min_{\gamma\geq0}\max_{\frac{1}{k-\floor{k/2}}\leq x\leq 1}f(\gamma, x)&\leq \max_{x\leq 1}h(x)\\
&=\max_{z\in\mathbb{Z}_{\geq 1}}g(z)\\
&=g(z^*)\\
&=h(x^*)\\
&=f\lrA{\ln\lrA{\frac{k-2z^*+(z^*)^2x^*+z^*x^*}{\floor{k/2}+1}}, x^*}\\
&=f\lrA{\ln\lrA{\frac{k-z^*+1}{\floor{k/2}+1}}, x^*}.
\end{align*}
Note that for any $k\in\mathbb{Z}_{\geq3}$, it is easy to verify that
\begin{align*}
\frac{k-z^*+1}{\floor{k/2}+1}&\geq\frac{k-\ceil{z_0}+1}{\floor{k/2}+1}\\
&=\frac{k-\Ceil{\frac{\sqrt{4k+5}-1}{2}}+1}{\floor{k/2}+1}\\
&\geq 1.
\end{align*}
Let $\gamma^*\coloneqq\ln\lrA{\frac{k-z^*+1}{\floor{k/2}+1}}$. Note that we have $\gamma^*=\ln\left(\frac{k+1-\ceil{z_0}}{\floor{k/2}+1}\right)$ if $g(\ceil{z_0})\geq g(\floor{z_0})$, and $\gamma^*=\ln\left(\frac{k+1-\floor{z_0}}{\floor{k/2}+1}\right)$ otherwise. Then, the LP-based tree-partition algorithm with $\gamma=\gamma^*$ has an approximation ratio of at most $\max\{g(\ceil{z_0}), g(\floor{z_0})\}$.

By the analysis of the previous two cases, we have
\[
\max_{0<x\leq \frac{1}{k-\floor{k/2}}}f(0, x)=3-\frac{1}{k-\floor{k/2}},
\]
and
\[
\max_{\frac{1}{k-\floor{k/2}}<x\leq 1}f(\gamma^*, x)\leq \max\{g(\ceil{z_0}), g(\floor{z_0})\}.
\]
For any $k\in\mathbb{Z}_{\geq3}$, it is easy to verify that
\[
3-\frac{1}{k-\floor{k/2}}\leq \max\{g(\ceil{z_0}), g(\floor{z_0})\}.
\]
Hence, by running the LP-based tree-partition algorithm with $\gamma=0$ and $\gamma=\gamma^*$ twice, the better found solution has an approximation ratio of $\max\{g(\ceil{z_0}), g(\floor{z_0})\}$.
\end{proof}
\endgroup

\section{Conclusion}
In this paper, we study approximation algorithms for $k$-MCVRP. Previously, only a few results were available in the literature. Based on recent progress in approximating $k$-CVRP, we design improved approximation algorithms for $k$-MCVRP. When $k$ is general, we improve the approximation ratio to $4-1/1500$ for splittable and unit-demand $k$-MCVRP and to $4-1/50000$ for unsplittable $k$-MCVRP; when $k$ is fixed, we improve the approximation ratio to $3+\ln2-\max\{\Theta(1/\sqrt{k}),1/9000\}$ for splittable and unit-demand $k$-MCVRP and to $3+\ln2-\Theta(1/\sqrt{k})$ for unsplittable $k$-MCVRP.

We remark that for unsplittable, splittable, and unit-demand $k$-MCVRP with fixed $3\leq k\leq 11$ the current best approximation ratios are still $H_k-\Theta(\ln^2 k/k)$~\cite{gupta2023local}. In the future, one may study how to improve these results.

A more general problem than $k$-MCVRP is called \emph{Multidepot Capacitated Arc Routing} (MCARP), where both vertices and arcs are allowed to require a demand. For MCARP, the current best-known approximation algorithms on general metric graphs are still based on the cycle-partition algorithm (see~\cite{yu2021approximation}).
Some results in this paper may be applied to MCARP to obtain some similar improvements.

\backmatter

%\bmhead{Supplementary information} If your article has accompanying supplementary file/s please state so here. Authors reporting data from electrophoretic gels and blots should supply the full unprocessed scans for key as part of their Supplementary information. This may be requested by the editorial team/s if it is missing. Please refer to Journal-level guidance for any specific requirements.

\bmhead{Acknowledgements} The work is supported by the National Natural Science Foundation of China, under grant 62372095.

\section*{Declarations}
\subsection*{Funding}
Not applicable.

\subsection*{Competing interests}
The authors declare no competing interests.

\subsection*{Ethics approval and consent to participate}
Not applicable.

\subsection*{Consent for publication}
Not applicable.

\subsection*{Data availability}
Not applicable.

\subsection*{Materials availability}
Not applicable.

\subsection*{Code availability}
Not applicable.

\subsection*{Author contribution}
All authors wrote the main manuscript text. All authors reviewed the manuscript.
\bibliography{main}
\end{document}